\documentclass[10pt]{article}

\usepackage[english]{babel}
\usepackage[utf8]{inputenc}
\usepackage{times}
\usepackage{latexsym}
\usepackage{amssymb}
\usepackage{amsmath,amsthm}
\usepackage{graphicx}
\usepackage{mathrsfs}
\usepackage{color}
\usepackage{multicol}
\usepackage{multirow}

\newtheorem{proposition}{Proposition}
\theoremstyle{remark}
\newtheorem{remark}{Remark}
\newtheorem{example}{Example}
\usepackage{comment}

\DeclareMathOperator*{\argmin}{arg\,min}

\title{Optimal trading policies for wind energy producer} 

\author{Zongjun Tan and Peter Tankov\\ Laboratoire de Probabilit\'es et Mod\`eles Al\'eatoires\\ Universit\'e Paris-Diderot} 
\date{}

\begin{document}

\maketitle

\begin{abstract}
We study the optimal trading policies for a wind energy producer
who aims to sell the future production in the open forward, spot,
intraday and adjustment markets, and who has access to imperfect
dynamically updated forecasts of the future production. We construct a
stochastic model for the forecast evolution and determine the optimal
trading policies which are updated dynamically as new forecast
information becomes available. Our results allow to quantify the
expected future gain of the wind producer and to determine the
economic value of the forecasts. 
\end{abstract}


\section{Introduction}
Wind power is now widely recognized as an important part of the
global energy mix, and the actors of the energy industry have no
choice but to cope with the intermittent and to a large extent
unpredictable nature of the wind power production. In particular, as
the guaranteed purchase schemes are either phased out or replaced with
more market-oriented subsidies, the wind power producers face the need
to sell the future power production in the open markets in the absence
of precise knowledge of the volume to be produced. The need of wind
power producers to adjust their delivery volume estimates as the
forecast becomes more precise is one of the factors behind the
development of intraday electricity markets, at which power can be
traded up to 45 minutes prior to delivery. 

The aim of this paper is to determine the optimal strategies for
selling the future power production of a single wind park for a wind producer who has access
to imperfect dynamically updated forecast of the future production,
which becomes progressively more precise as the production horizon
draws near. We formulate this problem as a stochastic optimization
problem where the power producer aims to maximize the expected gain
from selling electricity penalized by terms accounting for market
illiquidity and the extra cost of using the adjustment
market. To solve this problem, we develop a stochastic model for the
forecast evolution, and determine the optimal trading strategy which
is updated dynamically as new forecast information becomes available.
This allows to quantify the optimal expected gain for the producer,
and to compare the expected gain under different assumptions on the
forecast dynamics, thus quantifying the economic value of different
forecasts.

Wind power producers in Europe and in many other countries with
deregulated energy sector have access to four types of markets. 
\begin{itemize}
\item The forward market -- more than 1 day prior to delivery,
  delivery periods are day, week, month, quarter and year.
\item {Spot market} -- 1 day prior to delivery, delivery period
  is 1 hour or 30 minutes.
\item {Intraday market} -- between 1 day and 45 min,
  delivery period is 15 minutes.
\item {Adjustment} (imbalance) market (usually managed by the power
  network operator such as RTE in France) --
  the last 45 minutes. In the adjustment market, the bid-ask spread is
  very wide, which may be interpreted as a penalty imposed on the agents
 for using this market. 
\end{itemize}


Optimal trading strategies for wind power producer with a focus on
intraday markets have been
considered by several authors. Morales et al.~\cite{morales2010short}
consider the short-term trading for a wind power producer and
determine the optimal strategies starting from a small number of
scenarios of wind power production generated with an autoregressive
model, without taking into account the available forecasts. Henriot
\cite{henriot2014market} studies optimal design of intraday markets in
the presence of wind power producers who use certain pre-determined
strategies (without optimization). Garnier and Madlener
\cite{garnier2015balancing} show how forecast errors may be corrected
by optimal trading in intraday markets. The paper which is closest in
spirit to ours is A\"id et al.~\cite{aid2016optimal}. These authors
consider the optimal trading problem in intraday markets in the
presence of imperfect demand forecasts and market impact, however,
unlike our paper they do not focus on wind energy.

The rest of the paper is structured as follows. In section
\ref{production} we study the realized production data and show that
the distribution of the realized production is well described with a
truncated log-normal distribution. Section \ref{forecast} focuses on
forecast dynamics: using some ideas from financial mathematics, we
develop a stochastic model for the forecast evolution which is
compatible with the truncated log-normal distribution for the realized
production. Finally, in Section \ref{optimal}, we formulate and solve
in several different settings, relevant for large and small power producers,
the optimization problem for the wind power producer who aims to
maximize the expected gain from selling the future production.

\section{Modeling the realized production} \label{production}
We define the normalized output power of a wind park $F_T$ by 
$$
F_T = 0\vee \frac{P_T}{P_{\max}},
$$
where $P_T$ is the actual instantaneous power production (in practice
the instantaneous production will be replaced with 10-minute average),
and $P_{max}$ is the rated power of the park. Since some of the
turbine equipment consumes power, the actual realized power production
may sometimes have small negative values; to remove this effect, the
normalized power output is truncated from below by $0$. 

To build a model for the normalized output power, we assume that $F_T$
is obtained by applying a ``stylized power curve'' $f_{prod}$ to
the ``stylized wind speed'' $X_T$:
$$
F_T = f_{prod}(X_T)
$$ 
We emphasize that the model is built for
the output power directly and not for the wind; the power curve and
wind speed are introduced merely to provide a rationale for the
model. The stylized wind speed $X_T$ follows a log-normal
distribution with parameters $\mu_X$ and $\nu_X$, whose density is 
$$
\rho_X(x) =  \frac{1}{x \sqrt{2 \pi} \nu_X} \exp \left( \displaystyle - \frac{1}{2} \left( \frac{\ln x - \mu_X}{\nu_X} \right)^2 \right).
$$
We assume that the variable $X_T$ follows a log-normal distribution because:
\begin{itemize}
	\item The log-normal distribution has been used in the
          literature as a model for wind speeds
          \cite{garcia1998fitting}. It is also quite close to the Weibull
          distribution, which is the parametric model of choice for
          wind speed data; 
	\item The log-normal distribution is analytically tractable and allows to introduce a dynamical aspect into the model via a Brownian motion.
\end{itemize}

The stylized production function is  
\begin{equation*}
f_{prod}(x) = \frac{ ( x - x_{min} )^{+} - (x - x_{max} )^{+} }{x_{max}- x_{min}}.
\end{equation*}
This shape of this function is illustrated in Figure
\ref{prodfunc.fig}; note that there is no cut-out. 
\begin{figure}
\centerline{\includegraphics[width=0.6\textwidth]{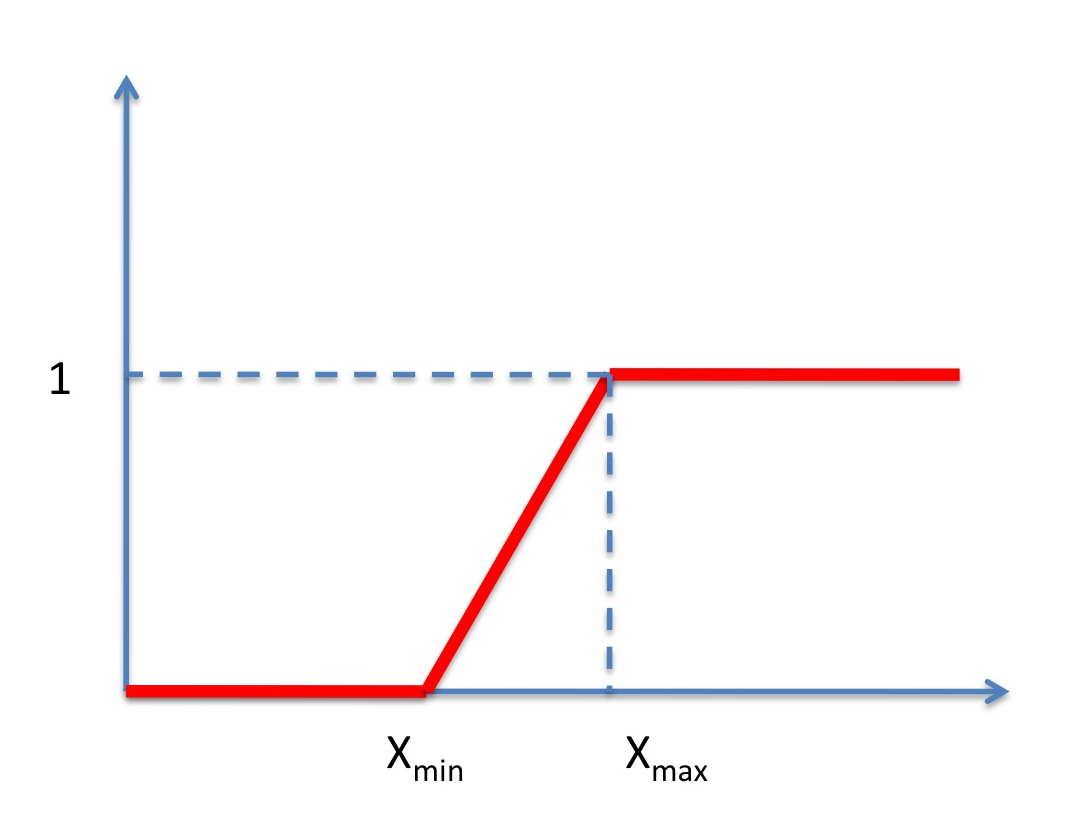}}
\caption{Stylized power curve used to model the realized production}
\label{prodfunc.fig}
\end{figure}

The above assumptions imply that $F_T $ follows a {truncated
  log-normal} distribution with parameters
\begin{equation}
	\left \{
	\begin{array}{l}
		\zeta = - \frac{x_{min}}{x_{max}-x_{min}} \\
		\mu  = \mu_{X} - \ln(x_{max}-x_{min})  \\
		\nu = \nu_{X}.
	\end{array}
	\right.
	\label{eq:3_parameter_representation}
\end{equation}
On the interval $(0,1)$ this distribution is absolutely continuous with
density given by 
\begin{equation}
	\left\{	
	\begin{array}{rcl}
	\rho_{F}(y|\mu,\nu,\zeta) &=&	\displaystyle  \frac{1}{ ( y - \zeta)  \sqrt{2 \pi} \nu} \exp \left( - \frac{(\ln( y - \zeta) - \mu)^2}{ 2 \nu^2}\right).
	\end{array}	
	\right.	
\end{equation}
In addition, at points $0$ and at $1$ the distribution has atoms given by
\begin{align*}
	\mathbb P[F_T=0] &=  \mathbb{P}(X_T \leq x_{min}) =
        \Phi\left(\frac{\ln x_{min} - \mu_{X}}{\nu_{X}} \right) =
        \Phi\left(\frac{\ln(-\zeta) - \mu}{\nu} \right):= P_0(\mu,\nu,\zeta)  \\
	\mathbb P[F_T = 1] &= \mathbb{P}(X_T > x_{max}) = 1 -
        \Phi\left(\frac{\ln(1 -\zeta) - \mu}{\nu} \right):= P_1(\mu,\nu,\zeta).  
\end{align*}
Note that while the original construction used four
parameters $(\mu_X,\nu_X,x_{min},x_{max})$, one parameter is
redundant, and the distribution of $F_T$ is completely characterized
by the three parameters $\mu,\nu,\zeta$. To remove this redundancy, we
shall set $\mu_X = -\frac{1}{2} \nu_{X}^2$ in the following, which
ensures that $\mathbb E[X_T]  =1$.

%

\paragraph{Fitting the model}

The model was fitted to the output power at the wind park level for 3 wind parks in
           France, sampled at 10-minute intervals from  Jan 1st,
           2011 to  Jan 1st, 2015, provided by the company Ma\"ia Eolis (hereafter
           referred to as Plant 1, Plant 2 and Plant 3). Figure
           \ref{realized.fig} shows the histograms of the realized
           production for the three plants (plants are numbered from
           left to right in this and other graphs).

\begin{figure}
	\centerline{\includegraphics[width=0.35\textwidth]{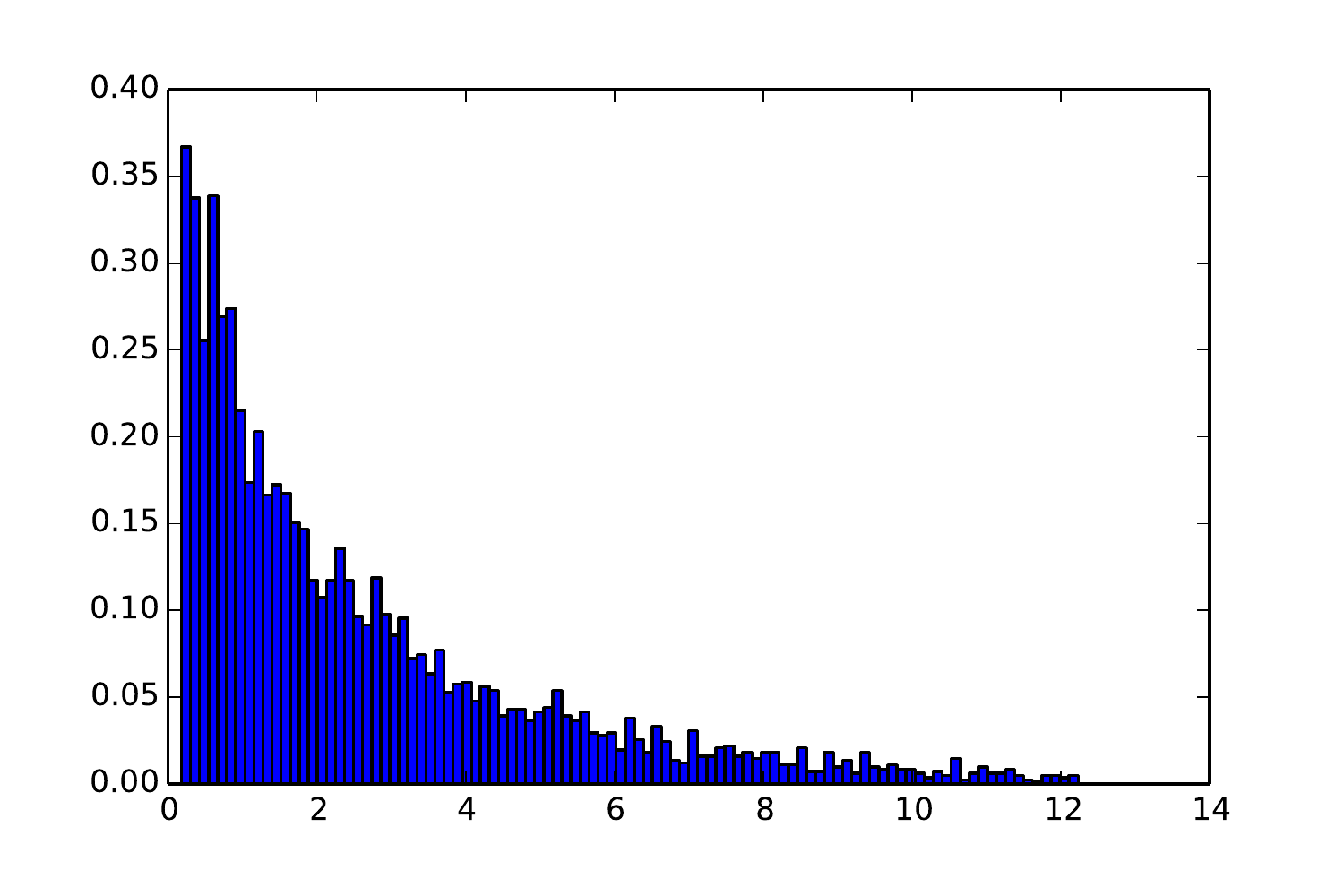}
	\includegraphics[width=0.35\textwidth]{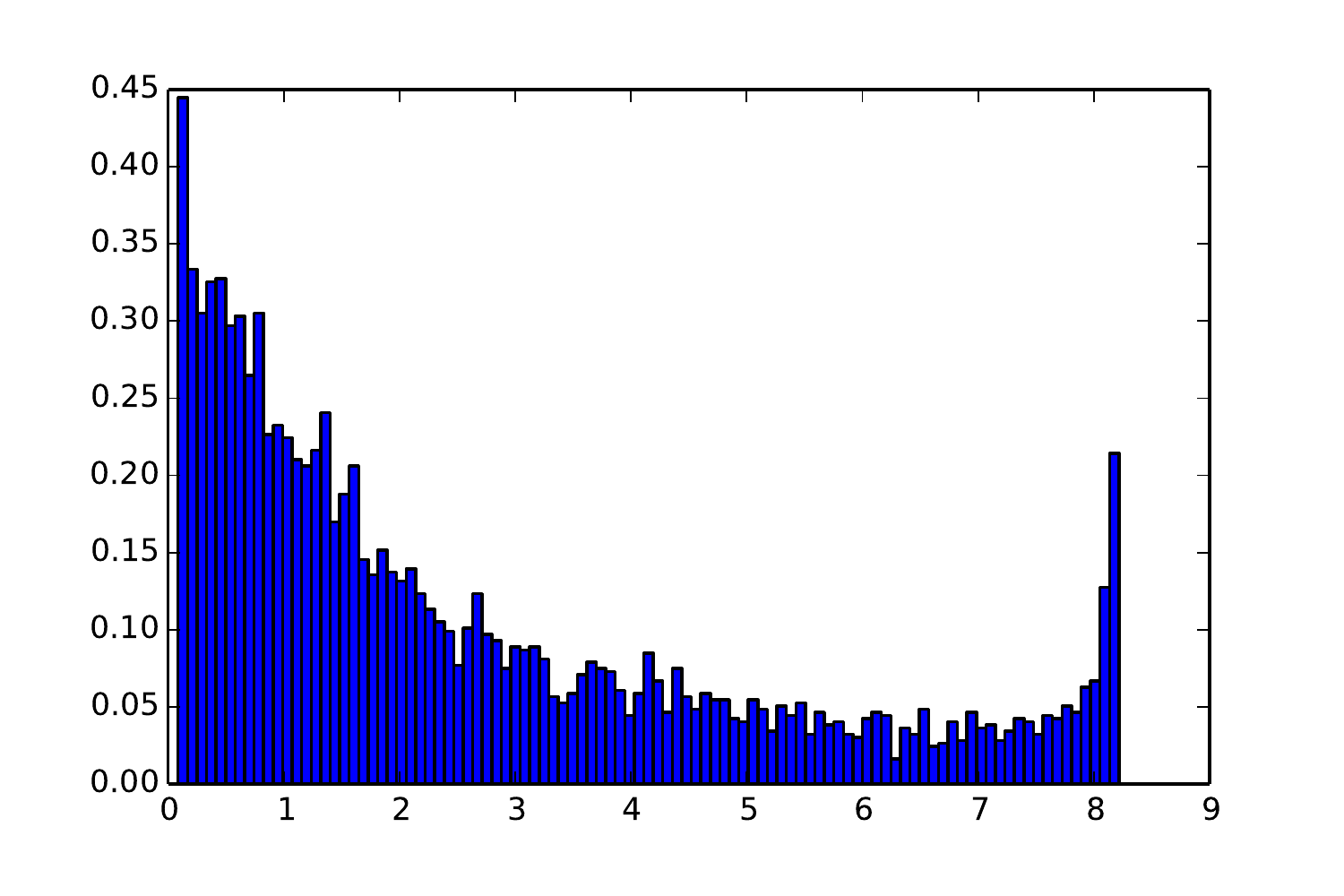}
	\includegraphics[width=0.35\textwidth]{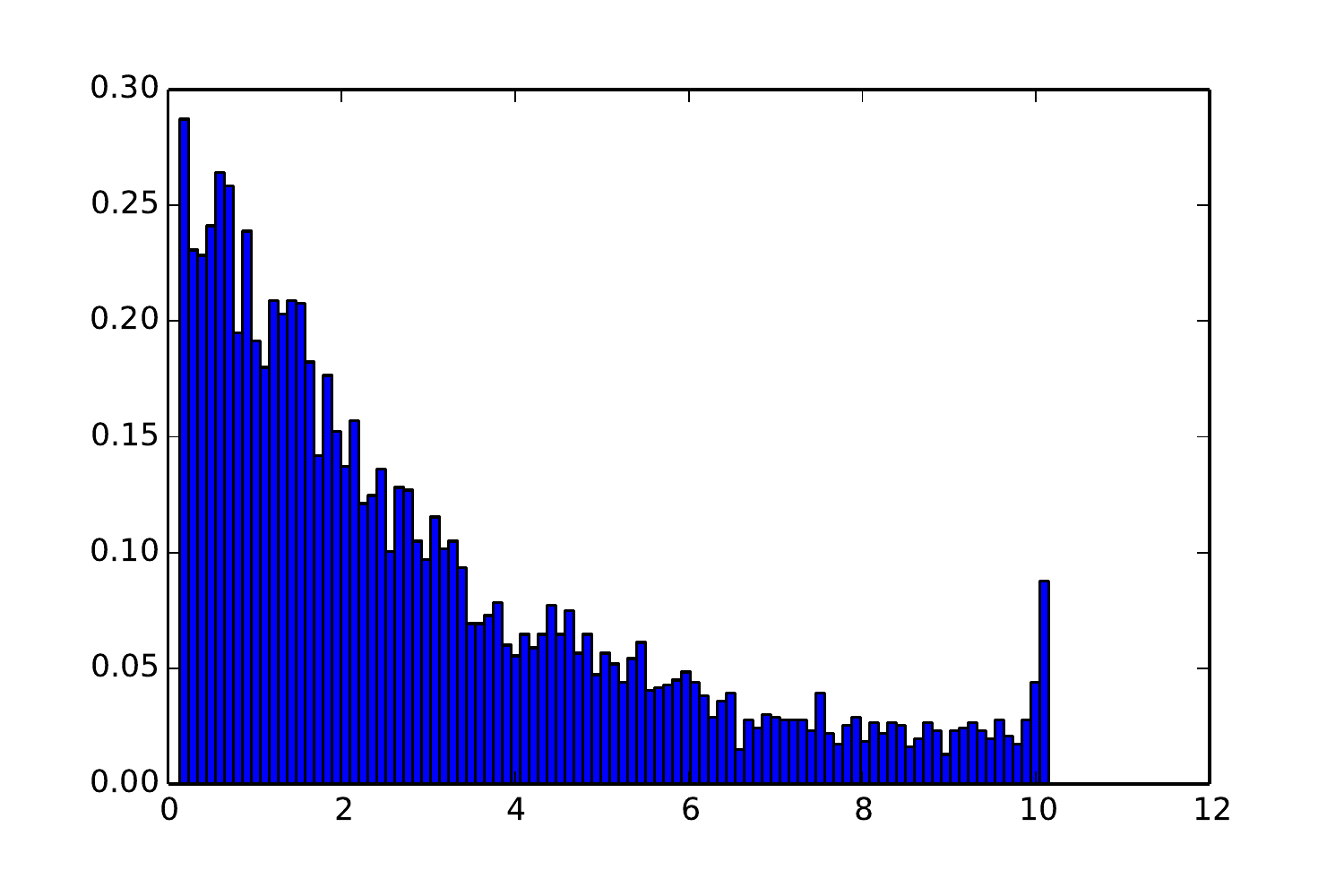}}
\caption{Histograms of 10-minute realized power production, with
  4-hour subsampling, for the three power
  plants which are the object of this study, excluding the atom at
  zero. }
\label{realized.fig}
\end{figure}

Denote the observed normalized output power values by
$(F_T^{k})_{k=1}^N$, and assume that they are arranged in
\emph{increasing order}. 
The method consists in minimizing the Euclidean distance between the
empirical quantiles and the quantiles of the theoretical
distribution. More precisely, given $\alpha \in [0,1]$, we define the
empirical quantile
\begin{equation}
q^\alpha_{emp}  = \max \left\{ F^k_{T} \Big| \frac{k}{N} \leq \alpha \right\},
\end{equation}
and, for $P_0(\mu,\nu,\zeta)\leq \alpha \leq 1-P_1(\mu,\nu,\zeta)$, we
define the theoretical quantile
$$
q^{\alpha}(\mu,\nu,\zeta) = \max \left\{ x \vert \Phi \left( \frac{ \ln(x - \zeta) - \mu}{\sigma} \right) \leq \alpha\right\},  
$$
where $\Phi$ is the standard normal distribution function. 
The parameters are estimated by minimizing 
$$
\sum_{l=1}^L \left( q^{\alpha_l(\mu,\nu,\zeta)}_{emp}- q^{\alpha_l(\mu,\nu,\zeta)}(\mu,\nu,\zeta) \right)^2,
$$
where $(\alpha_l)_{l=1}^L$ are probability levels, uniformly spaced
between  $P_0(\mu,\nu,\zeta)$ and  $1-P_1(\mu,\nu,\zeta)$, that is,
$$
\alpha_l = P_0(\mu,\nu,\zeta) + \frac{l-1}{L} (1 - P_1(\mu,\nu,\zeta)
- P_0(\mu,\nu,\zeta)). 
$$
In the numerical example below, $L=100$ probability levels were used. 

Table \ref{table:production_calibration} gives the fitted optimal
parameters $(\mu^*, \nu^*, \zeta^*)$ and the corresponding latent
parameters $(\mu_{X_T}, \nu_{X_T}, x_{min}, x_{max})$ obtained for the
three power plants. 
\begin{table}
	\centering
	\begin{tabular}{l|c|c|c}
		\hline
		{Parameters} & Plant 1 & Plant 2 & Plant 3\\
		\hline \hline
		$\mu$ &$-1.46551$&$-0.60213$&$-0.76199$\\
		$\sigma$ &$0.66020$&$0.46158$& $0.48778$\\
		$\zeta$ &$-0.13248$&$-0.33757$& $-0.26449$\\
		\hline \hline
		$x_{min}$ &$0.46129$&$0.55412$&$0.50312$\\
		$x_{max}$ &$3.94322$&$2.19561$&$2.40534$ \\
		$\mu_{X_T}$ &$-0.21793$&$-0.10653$&$-0.11896$\\
		$\sigma_{X_T}$ &$0.66020$&$0.46158$&$0.48778$ \\
		\hline
		
	\end{tabular}
	\caption[Parameters for power production]{Fitted parameters of
          normalized production and the corresponding parameters
          $(x_{min},x_{max},\mu_{X_T}, \sigma_{X_T})$.}
	\label{table:production_calibration}
\end{table}
The fitted truncated log-normal densities are
shown in Figure \ref{fittedprod.fig}. 

\begin{figure}
	\centerline{\includegraphics[width=0.35\textwidth]{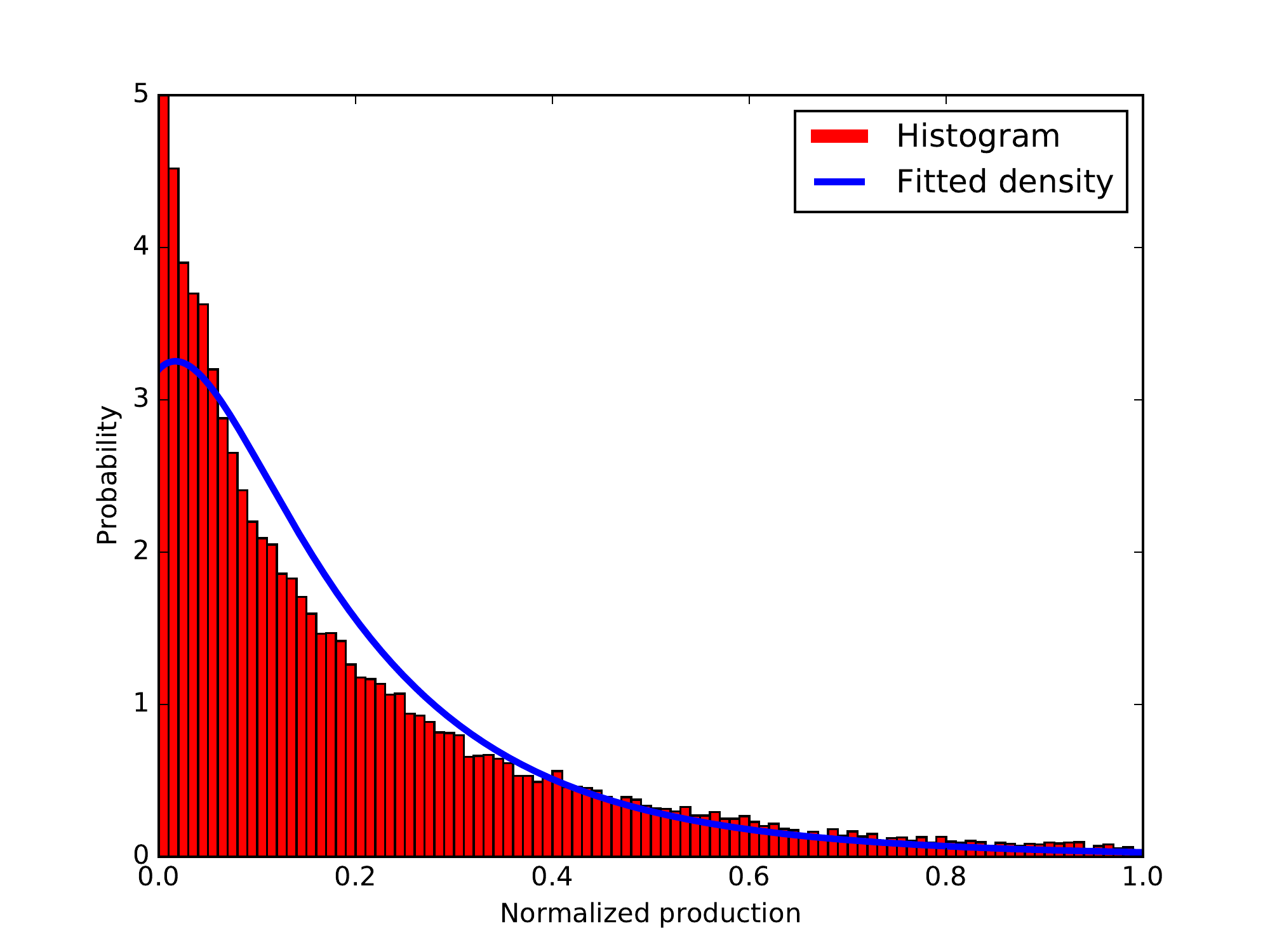}
	\includegraphics[width=0.35\textwidth]{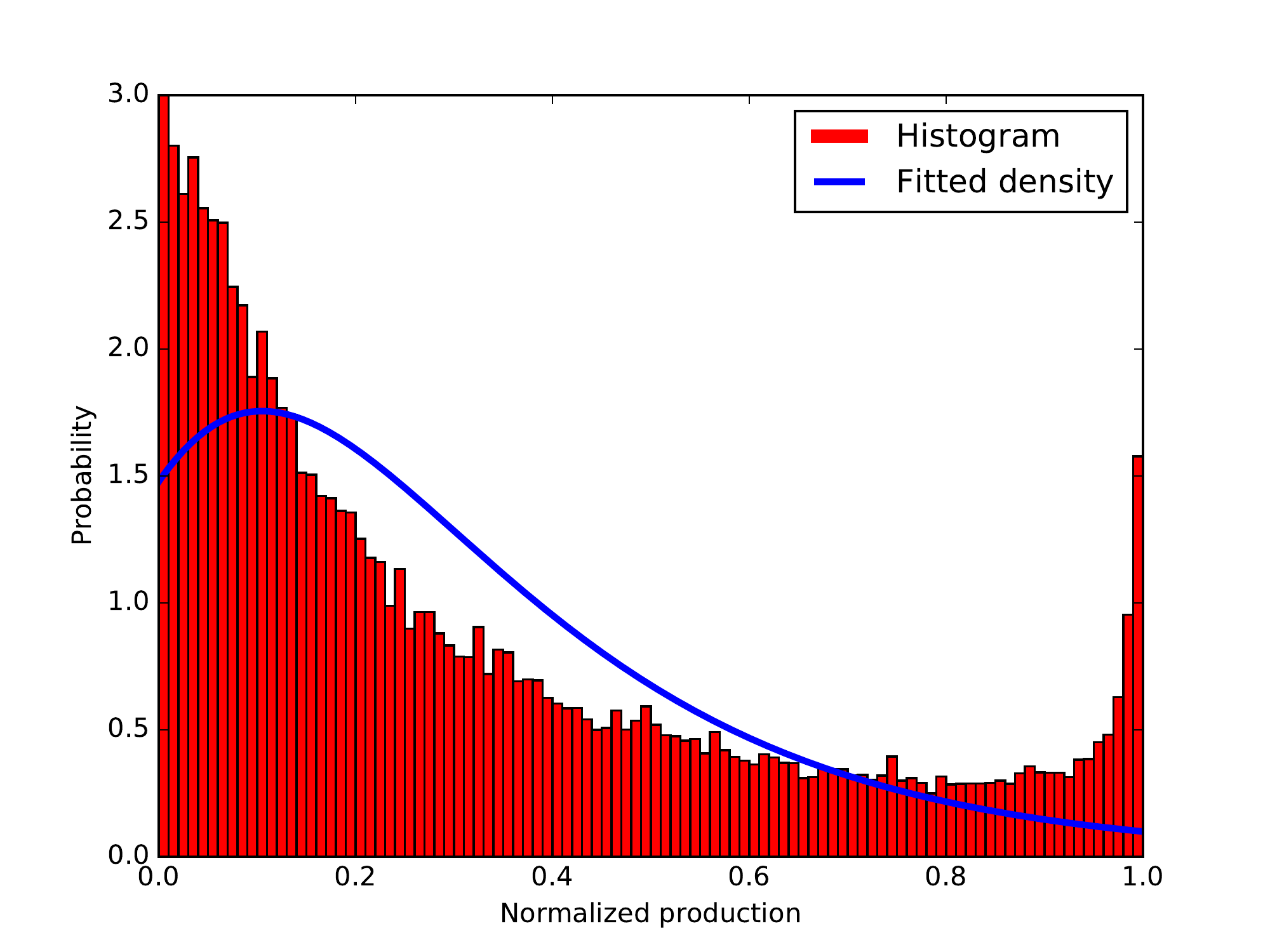}
	\includegraphics[width=0.35\textwidth]{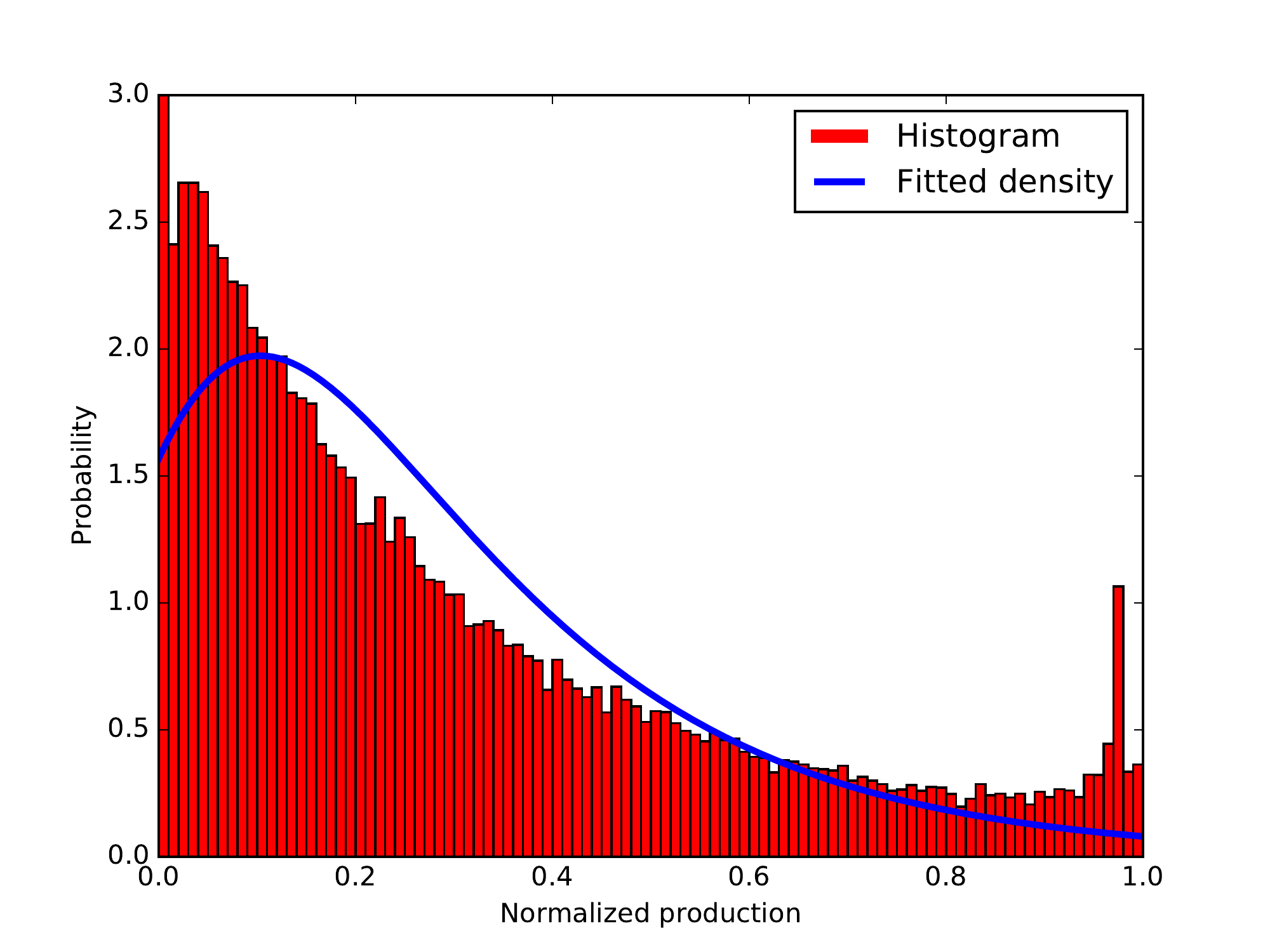}}
\caption{Fitted vs.~empirical densities for the three power
  plants. The total production has been normalized to one.}
\label{fittedprod.fig}
\end{figure}

\section{Modeling the forecast dynamics}\label{forecast}

To understand how to optimally update the trading strategy depending on
  the forecast, we need to build a dynamic stochastic model for the forecast,
  which is consistent with the distribution of the realized production
  described in the previous section. More precisely, at every date $t$
  we assume that the forecast $F_t$ is the best prediction of the
  realized production given the available information. 

To build a forecast model formalizing this idea, we need to define
 a filtration $\mathbb F = (\mathcal F_t)_{0\leq t\leq T}$, where
  $\mathcal F_t$ models the information available to the wind producer
  at time $t$, and a stochastic process $(F_t)_{0\leq t\leq T}$ with the
  following properties:
\begin{itemize}
\item It is a martingale with respect to the filtration $\mathbb F$;
\item $F_T$ has the truncated log-normal distribution described in the
  preceding section. 
\end{itemize}

We now proceed with the construction of the filtration and the process
$F$. Let $W$ be a standard Brownian motion, and $Z$ be a standard
normal random variable independent from $W$. We define the process
$(X_t)_{0\leq t\leq T}$ by 
$$
X_t  = \exp\left(\int_0^t \sigma(s) dW_s - \frac{1}{2}\int_0^t
  \sigma^2(s) ds \right), t<T
$$
and 
$$
X_T = \exp\left(\int_0^T \sigma(s) dW_s - \frac{1}{2}\int_0^T
  \sigma^2(s) ds \right) e^{bZ - \frac{b^2}{2}}. 
$$
where $(\sigma(s))_{0\leq s\leq T}$ is a square integrable
deterministic function and $b\geq 0$. We then define $\mathbb F$ to be
the natural filtration of $X$ completed with the null sets. 

In other
words, for each fixed $t$,
$$
X_T  \stackrel{d}{=} X_t e^{\sqrt{\theta(t)} N - \frac{\theta(t)}{2}}\quad
\text{where}\quad N\sim N(0,1)\quad \text{and}\quad \theta(t) =
\int_t^T \sigma(s)^2 ds + b^2,
$$
and $N$ is independent from $X_t$. Letting $(\mathcal F_t)_{0\leq t
  <T}$ be the completed natural filtration of the Brownian motion $W$
and $\mathcal F_T:= \mathcal F_0 \vee \sigma(Z)\vee \sigma(W_s, 0\leq s \leq T)$, 
we see that $X$ is an $\mathbb F$-martingale
which means that the variable $X_t$ may be seen as the best prediction
of the stylized wind $X_T$ given the information available at time
$t$. The jump at time $T$ is needed to model
the component of the wind which is not predictable even at very short
time horizons. It is clear that by taking 
$$
\nu_X = \theta(0),
$$
we recover the distribution of $X_T$ described in the preceding
section. 

We then define the \emph{forecast process} by
$$
F_t = \mathbb E[f_{prod}(X_T)|\mathcal F_t]. 
$$
The following proposition gives an explicit form of this process. 
\begin{proposition}
The forecast process is given explicitly by
$$
F_t = g(X_t,\theta(t)),
$$
where
\begin{multline*} g(x,\theta) = \frac{1}{x_{max} - x_{min}} \big[ x (
  \Phi(d^{min}_{+}(x,\theta))-\Phi(d^{max}_{+}(x,\theta)))\\- x_{min}
  \Phi(d^{min}_{-}(x,\theta))  + x_{max} \Phi(d^{max}_{-}(x,\theta))
\big]
\end{multline*}
with
$d_{\pm}^{min,max}(x,\theta) = \frac{1}{\sqrt{\theta}} [ \ln(x/x_{min,
  max}) \pm \theta / 2]$ 
and $\Phi$ is the standard normal distribution function. 
\end{proposition}

This model fully describes the evolution of the forecast dynamics,
while ensuring that $F_t(T) \in [0,1]$ for all $t$. For every forecast
horizon, the forecast distribution is parameterized by a single
number, $\theta(t)$. Since the key quantity for determining the optimal
strategy is the forecast error, we fit the function $\theta$ by
matching the empirically observed variances of the
forecasting errors for different horizons with the variances predicted by the model
and given by
\begin{align*}
\mathbb E\left[(F_t-F_T)^2\right] &= \mathbb
E\left[(g(\theta(t),X_t)-f_{prod}(X_T))^2\right]\\
& = \mathbb E\left[f_{prod}(X_T))^2\right]-\mathbb
E\left[g(\theta(t),X_t)^2\right].
\end{align*}
Computing the second term requires a one-dimensional numerical
integration:
$$
\mathbb
E\left[g(\theta(t),X_t)^2\right] = \frac{1}{\sqrt{2\pi}}\int_{\mathbb R} g\left(\theta(t),
\exp\left(\sqrt{\nu_X^2 - \theta(t)} z - \frac{\nu_X^2 - \theta(t)}{2}\right)\right)^2
e^{-\frac{z^2}{2}} dz. 
$$
As for the first term, it may be evaluated explicitly:
\begin{align*}
\mathbb E\left[f_{prod}(X_T))^2\right] &= \Phi\left(d_-^{max}\right) +
\frac{e^{\frac{\nu_X^2}{2}}}{(x_{max}-x_{min})^2}
\left\{\Phi\left(d^{min}_0\right) - \Phi\left(d^{max}_0\right)
\right\} \\ &- \frac{2 x_{min}}{(x_{max}-x_{min})^2}
\left\{\Phi\left(d^{min}_+\right) - \Phi\left(d^{max}_+\right)
\right\} \\ &+ \frac{x_{min}^2}{(x_{max}-x_{min})^2}
\left\{\Phi\left(d^{min}_-\right) - \Phi\left(d^{max}_-\right)
\right\},
\end{align*}
where 
\begin{align*}
&d_0^{max,min}=\frac{-\log
    x_{max,min} + \frac{3\nu_X^2}{2}}{\nu_X} ,\quad d_\pm^{max,min}=\frac{-\log
    x_{max,min} \pm \frac{\nu_X^2}{2}}{\nu_X}.
\end{align*}

Therefore, for fixed parameters of the realized production
distribution $\nu_X, x_{min}$ and $x_{max}$, the forecast error variance $\mathbb
E\left[(F_t-F_T)^2\right]$ is a function of $\theta(t)$ only. For
$\theta \in [0,\nu_X^2]$, let
$$
\phi(\theta) = \mathbb E\left[f_{prod}(X_T))^2\right]-\frac{1}{\sqrt{2\pi}}\int_{\mathbb R} g\left(\theta,\exp\left(\sqrt{\nu_X^2 - \theta} z - \frac{\nu_X^2 - \theta}{2}\right)\right)^2
e^{-\frac{z^2}{2}} dz .
$$
By Jensen's inequality it can be shown that $\phi(\theta)$ is strictly
increasing in $\theta$. Moreover it is clearly continuous and
satisfies $\phi(0) = 0$ and $\phi(\nu_X^2) = \text{Var}\,
[f_{prod}(X_T)]$. Therefore, for any $v$ in the interval $(0,\text{Var}\,
[f_{prod}(X_T)])$, there exists a unique $\theta$ such that
$\phi(\theta)=v$. We use this property to calibrate the function
$\theta(\cdot)$ non-parametrically to the observed variances of the
forecast errors. 





Alternatively, one can use a parametric volatility function given by
$$ 
\sigma_t = \sigma_0 e^{ \eta ( T - t) } \mathbf{1}_{ t > T - \tau^*}.
$$
Here, $\tau^*$ is the time horizon for which the forecast error
variance becomes equal to the unconditional variance of the realized
power output, which means that the forecast becomes useless. This
corresponds to the function $\theta(t)$ given by
$$
\theta(t) =\left\{ b^2 + \frac{\sigma^2_0}{2 \eta} \left(  e^{2 \eta (T - t) } - 1
\right)\right\}\wedge \nu_X^2. 
$$

\paragraph{Fitting the model}



	
%

We estimate the function $\theta(\cdot)$ in both the non-parametric
and the parametric form using the forecast data provided by Ma\"ia
Eolis. This data set contains the forecasts of the power output at
wind park level, produced by an independent forecasting company, for
the period from December 7th 2011 to March 3rd 2015. In the numerical
examples we focus on the wind park 1 from the three parks
considered in the previous section. The forecasts are
updated every $6$ hours and cover time horizons from 1h15min to 144
hours ahead with 15 minute step. The forecast values are positive, and
in the analysis we normalize them by the rated power of the plant so
that $F_t(T)\in [0,1]$. Figure \ref{forecast.fig} shows examples of
forecasts together with the actual realized production. The forecasts
appear quite precise for short time horizons, but the precision
deteriorates significantly for longer horizons. This is further
confirmed in Figure \ref{forecast_error.fig} which shows the
histograms of the forecast errors for different horizons.

\begin{figure}
\begin{center}
		\includegraphics[width=0.8\linewidth]{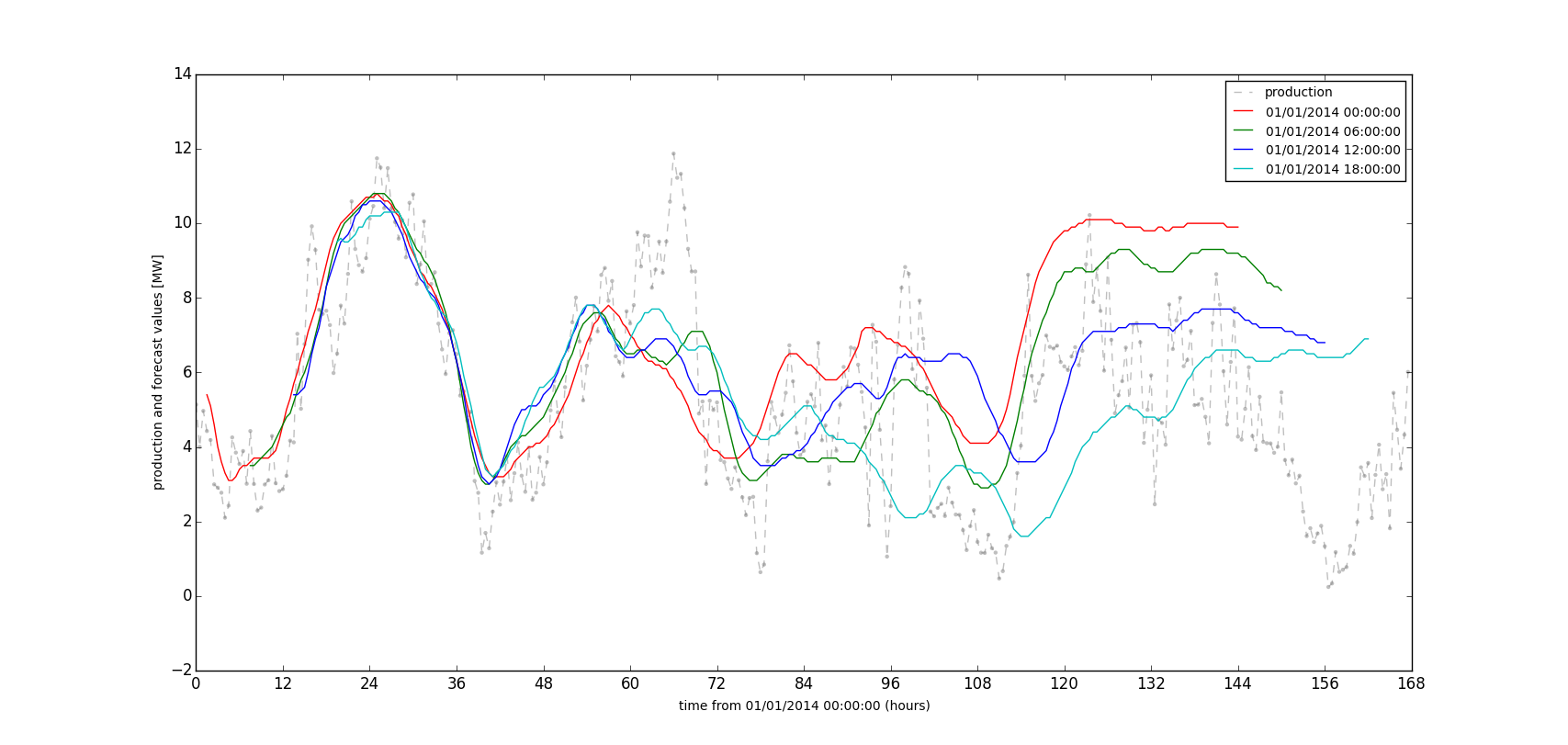}
\end{center}
\caption{Plot of the forecast made at a given date as function of time horizon
together with the realized production for this horizon for four
different starting times (given in the legend).	Accuracy
decreases for longer horizons.  }
\label{forecast.fig}
\end{figure}

\begin{figure}
\begin{center}
		\includegraphics[width=0.6\textwidth]{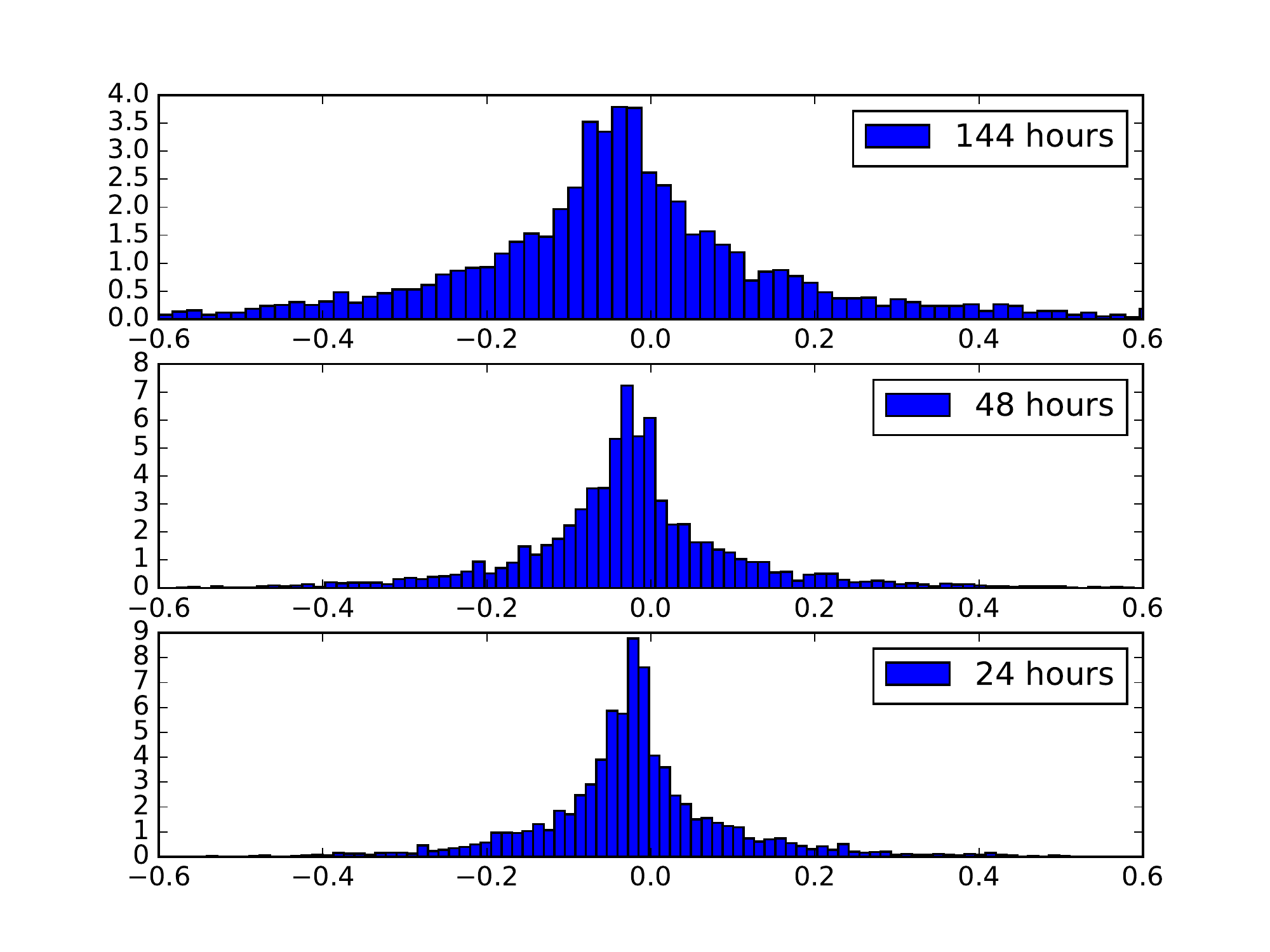}
\end{center}
\caption{Histograms of the forecast error for different time horizons.}
\label{forecast_error.fig}
\end{figure}

Figure \ref{estimation_theta.fig}, left graph, plots the variance of
the forecast error as function of time horizon. More than $\tau^* =120h$
prior to production date the variance of the forecast error exceeds
that of the realized production and we consider that the forecast has
no value.  The right graphs of
this figure shows the function $\theta$ estimated using the
non-parametric method described above and the parametric method
(when all error variances are fitted at the same time by nonlinear
least squares). The estimated parameter values are $\sigma_0 =
0.040113$, $\eta = 0.004423$ and $b = 0.308817$.

Finally, Figure \ref{model_histogram.fig} compares the empirical
distribution of the forecast error with the one generated by the model
for the time horizon of $48$ hours. 

\begin{figure}
\centerline{\includegraphics[width=0.45\textwidth]{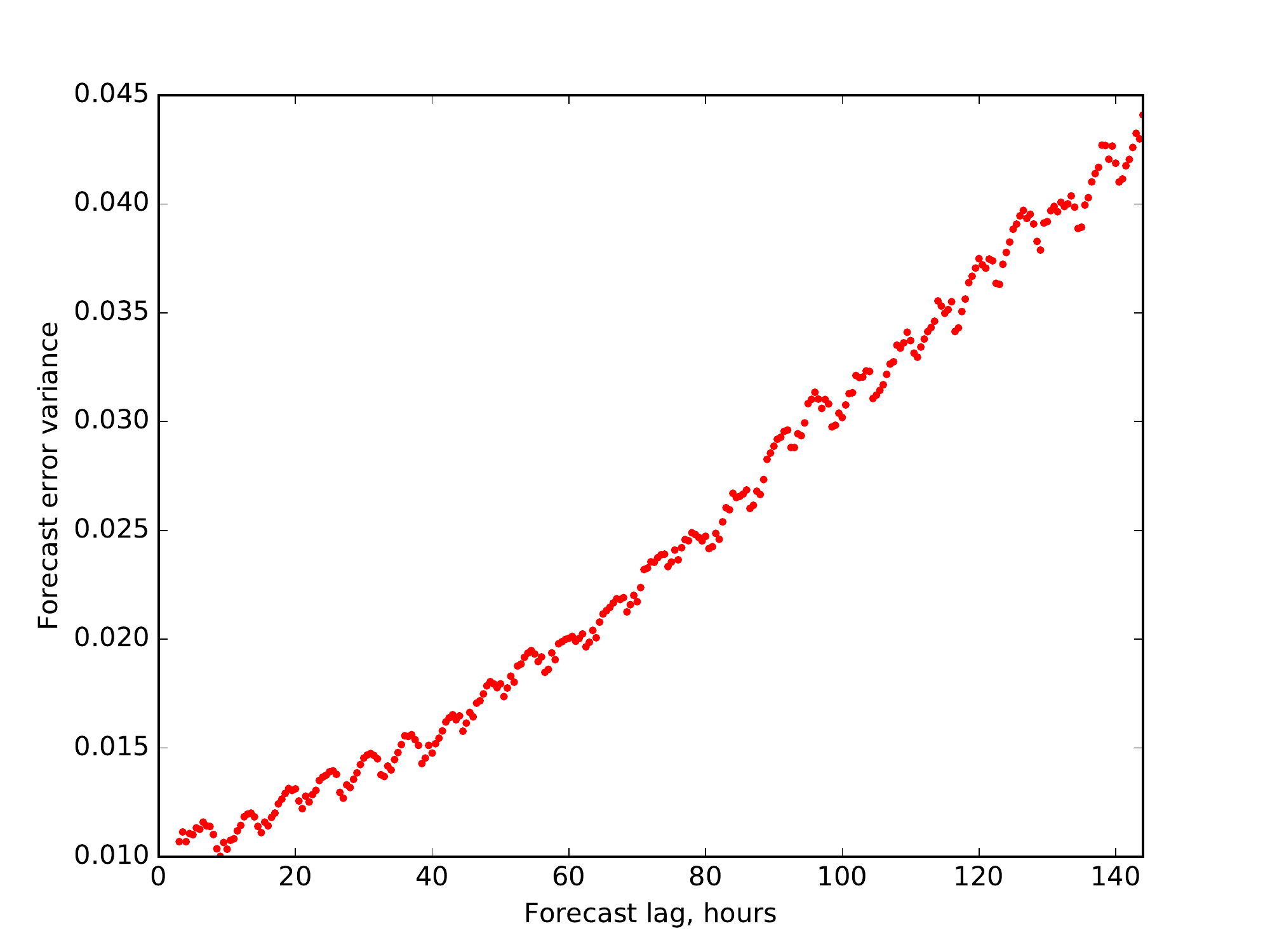}
   \includegraphics[width=.45\linewidth]{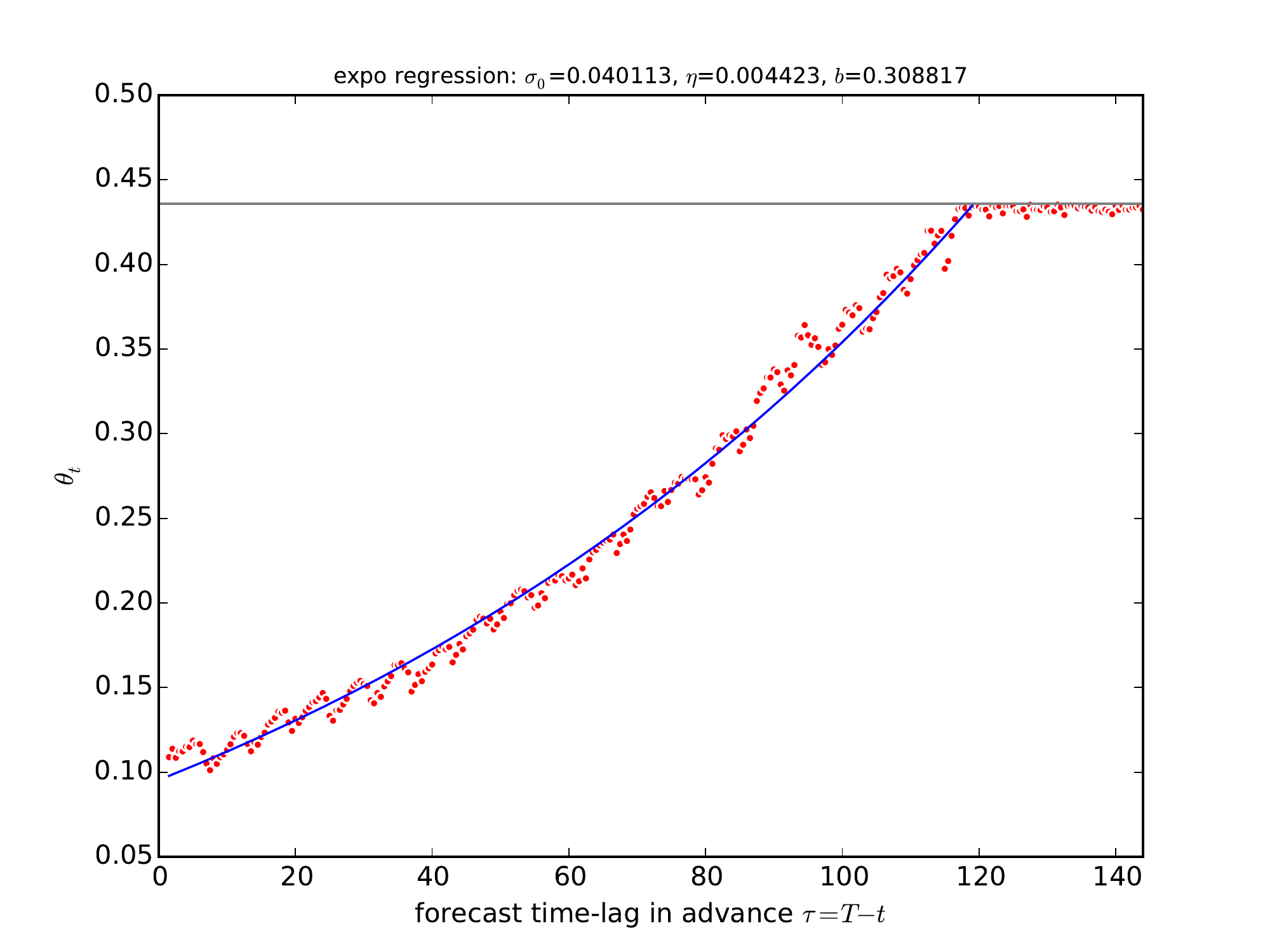}}
\caption{Left: variance of
the forecast error as function of time horizon. Right: function $\theta$
estimated using parametric and nonparametric method. The 6-hour
periodicity is due to the fact that forecast is updated every 6 hours.}
\label{estimation_theta.fig}
\end{figure}

\begin{figure}
					
\centerline{\includegraphics[width=.45\linewidth]{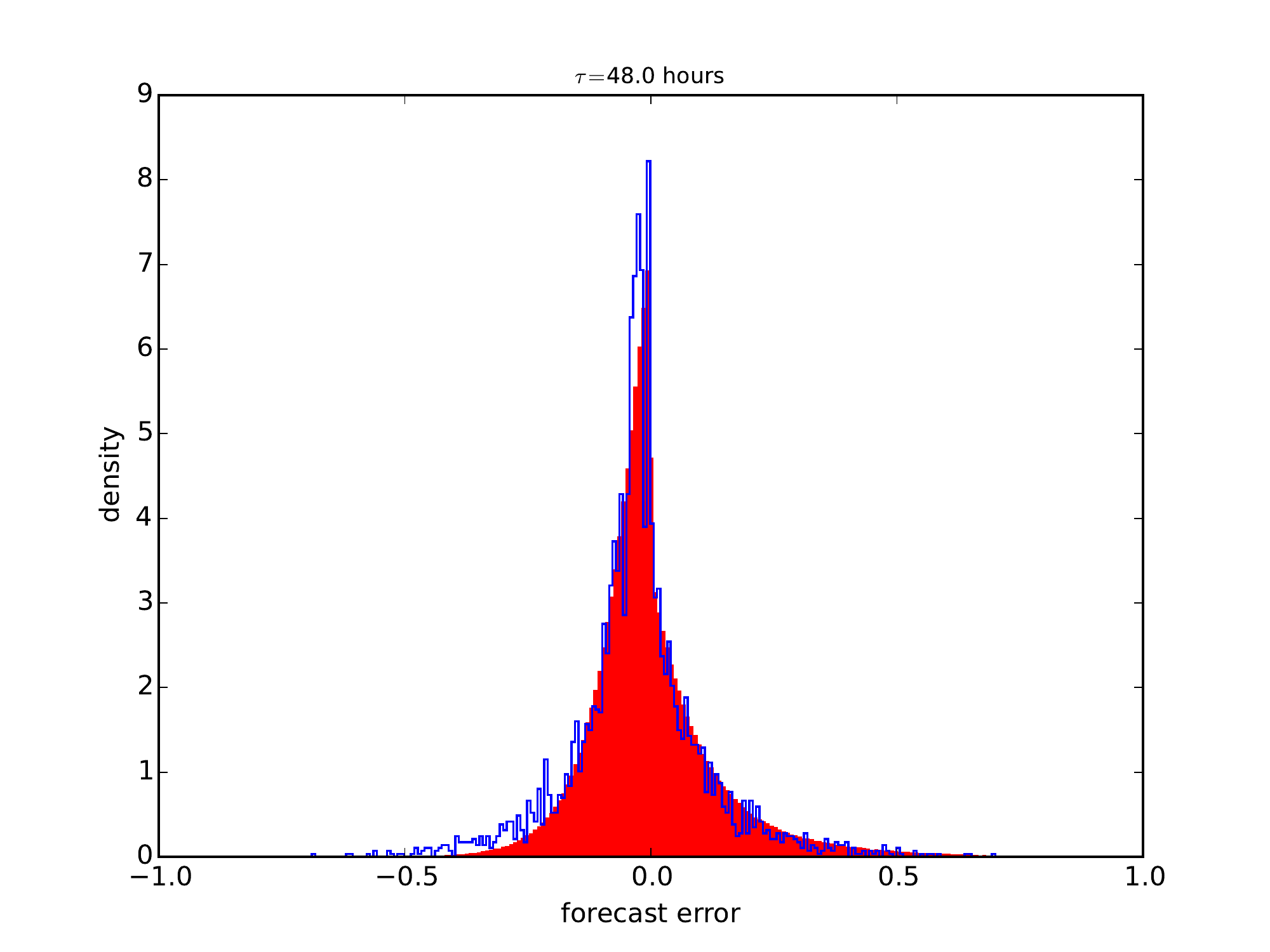}}
\caption{The histogram of the observed forecasting error (blue line) compared to
  the model-generated histogram for the time horizon of $48$ hours
  (red bars).}
  \label{model_histogram.fig}   
\end{figure}

\section{Optimization of market interventions}\label{optimal}
In this section, our aim is to determine the optimal strategies for
selling electricity produced during a short
  time period $[T-\delta,T]$, where $T$ is fixed.  This electricity must be sold in advance, in different markets
  (spot, forward, intraday), otherwise a {penalty} is applied
  for using the adjustment market. We assume that the wind power producer does not know the exact production but
  has {forecasts available.}

There are several reasons for trading both in intraday and
spot/forward markets. First of all, intraday markets are very volatile 
and illiquid: if supply exceeds demand, the prices plunge down and
if demand exceeds supply, the prices shoot up (see Figure
\ref{intraday.fig}).  This means that large amounts of
energy can only be sold at a very low price. For this reason, it is
advantageous to sell in spot / forward markets if the amount of
electricity to be produced is known in advance.
Also, by selling in the forward market, one {reduces
  the risk} associated to the change in the price until the delivery
  date, since forward prices fluctuate less than spot / intraday
  prices. On the other hand, selling in the spot/intraday market {reduces the
penalty} applied for not delivering the right amount since the
forecasts are better when the delivery horizon is close.

\paragraph{Forward price model}
In practice, the forward contracts are traded continuously but cover
an extended delivery period, e.g., one year, one quarter, one month,
one week and sometimes one day. In the spot market, trading takes
place only once per day, and one can make separate bids for each hour
of the following day. The intraday market again allows continuous
trading and the basic contract covers a $15$-minute delivery
period. To simplify the treatment and make our main ideas transparent,
we do not distinguish between different markets, and assume that at
every time $t$, one can enter into a forward contract allowing to buy
/ sell electricity at a future date $T$ at the price $P_t(T)$. Our
methods and results can be easily adapted to a more realistic market
structure. 

Let $\mathbb G:= (\mathcal G_t)_{t\geq 0}$ be the filtration of the
agent selling electricity. 
We assume that the forward sale price process satisfies 
$$
dP_t(T) = \mu_t dt + \beta_t dB_t,
$$
where $\mu$ and $\beta$ are deterministic processes such that
$$
\int_0^T (|\mu_t| + \beta_t^2) dt <\infty
$$ and $B$ is a
$\mathbb G$-Brownian motion. For longer
horizons, the coefficient
$\mu$ reflects the average trend of forward prices as the delivery
horizon draws near. As seen from Figure \ref{future.fig}, this trend
is typically negative, which corresponds to a premium for early
trading. For shorter time horizons the negative coefficient $\mu$ may
reflect the widening of the bid-ask spread in the intraday market.

\paragraph{Volume penalty} We assume that the wind power producer has the obligation to sell all the
produced energy and denote by {$\phi_t$ the aggregate
position} at time $t$ (total quantity to deliver at time $T$ owing to
the contracts entered into prior to date $t$).  The trading starts
at some fixed date $0$. If, at date $T$, $\phi_T \neq F_T$,
the agent must sell / purchase the extra energy at price $P_T:=P_T(T)$, and in
addition pay a penalty equal to $u(F_T - \phi_T)$, where $u(0)=0$,
$u(x)$ is increasing for $x>0$ and decreasing for $x<0$.

\paragraph{Admissible strategies}We are interested in determining the optimal strategies for two kinds
of electricity producers: a small producer whose interventions do not
affect market prices, and a relatively large producers whose trades
may impact the market. The small producer is only selling the electricity and does not engage in
proprietary trading. Therefore, the class $\mathcal A$ of admissible
strategies for a small producer contains all $\mathbb G$-adapted
increasing processes $\phi$ with $\phi_0=0$ satisfying the condition
$$
\mathbb E\left[\int_0^T (\phi_t |\mu_t| + \phi_t^2 \beta^2_t ) dt\right] <\infty.
$$
Indeed, allowing
$\phi$ to both increase and decrease does not make sense in the
absence of market impact since in that case the optimal strategy would
be to sell all produced electricity just before the terminal date.

For
the large producer, following \cite{almgren2001optimal},
we assume that the trading strategy is absolutely continuous and 
introduce a market impact term proportional to the square of the rate of trading $\psi_t = \phi'_t$.
The class  $\mathcal A^{ac}_+$ of admissible strategies for a large
producer who can only sell electricity thus contains all processes in $\mathcal A$ which
are absolutely continuous. Finally, the class $\bar A^{ac}$ of admissible strategies for a large
producer who can both buy and sell electricity contains all
processes of the form $\phi = \phi^+ - \phi^-$ where $\phi^+$ and
$\phi^-$ are in $\mathcal A^{ac}_+$.

\begin{figure}
\centerline{
\includegraphics[width=0.5\textwidth]{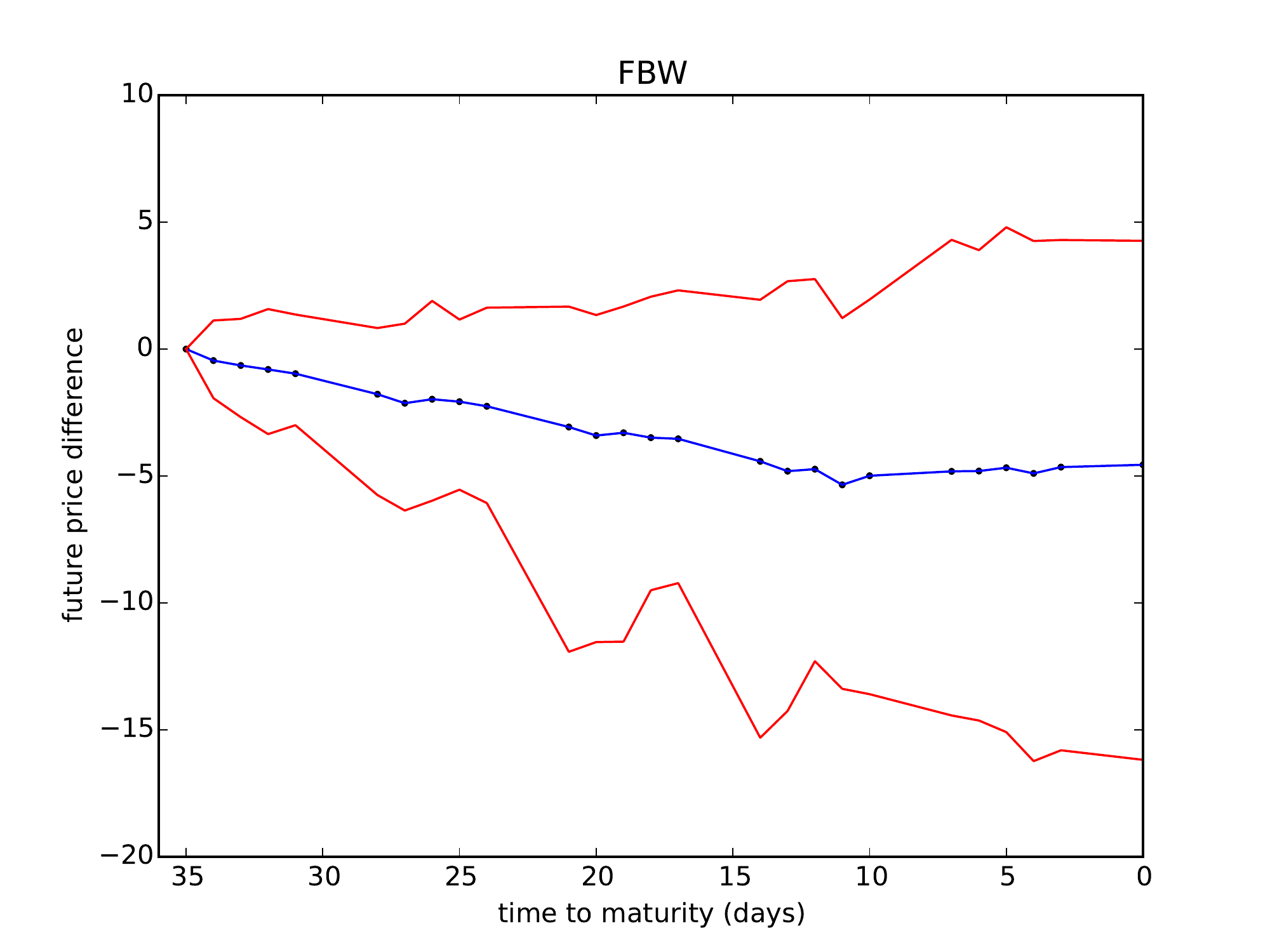}\includegraphics[width=0.5\textwidth]{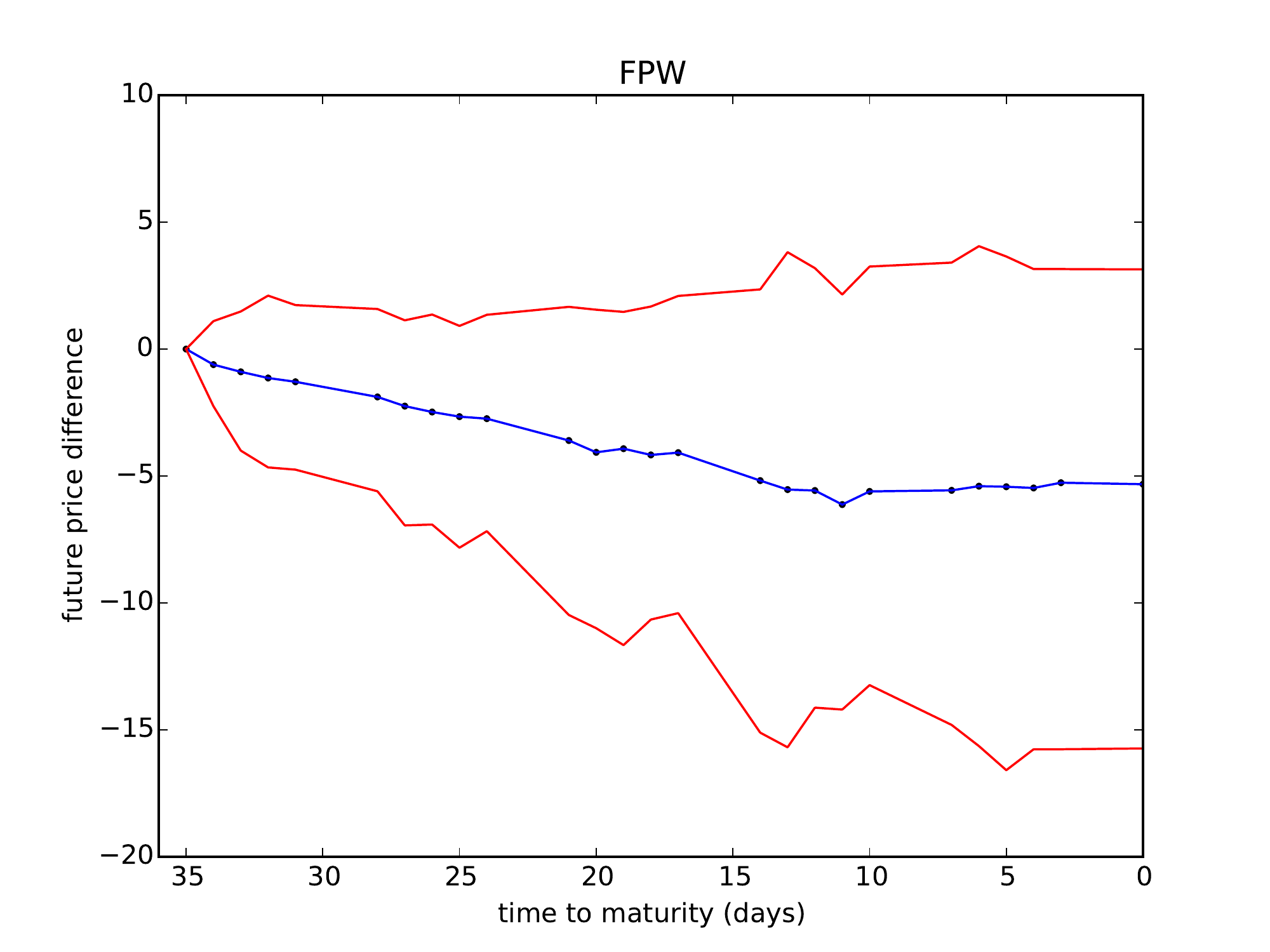}}

\caption{Evolution of future price difference $P_t(T)-P_0(T)$ as function of
$t$, averaged over one year, with $95\%$ confidence bounds. Left: base futures. Right: peak futures.}
\label{future.fig}
\end{figure}

\begin{figure}
\centerline{
\includegraphics[width=0.5\textwidth]{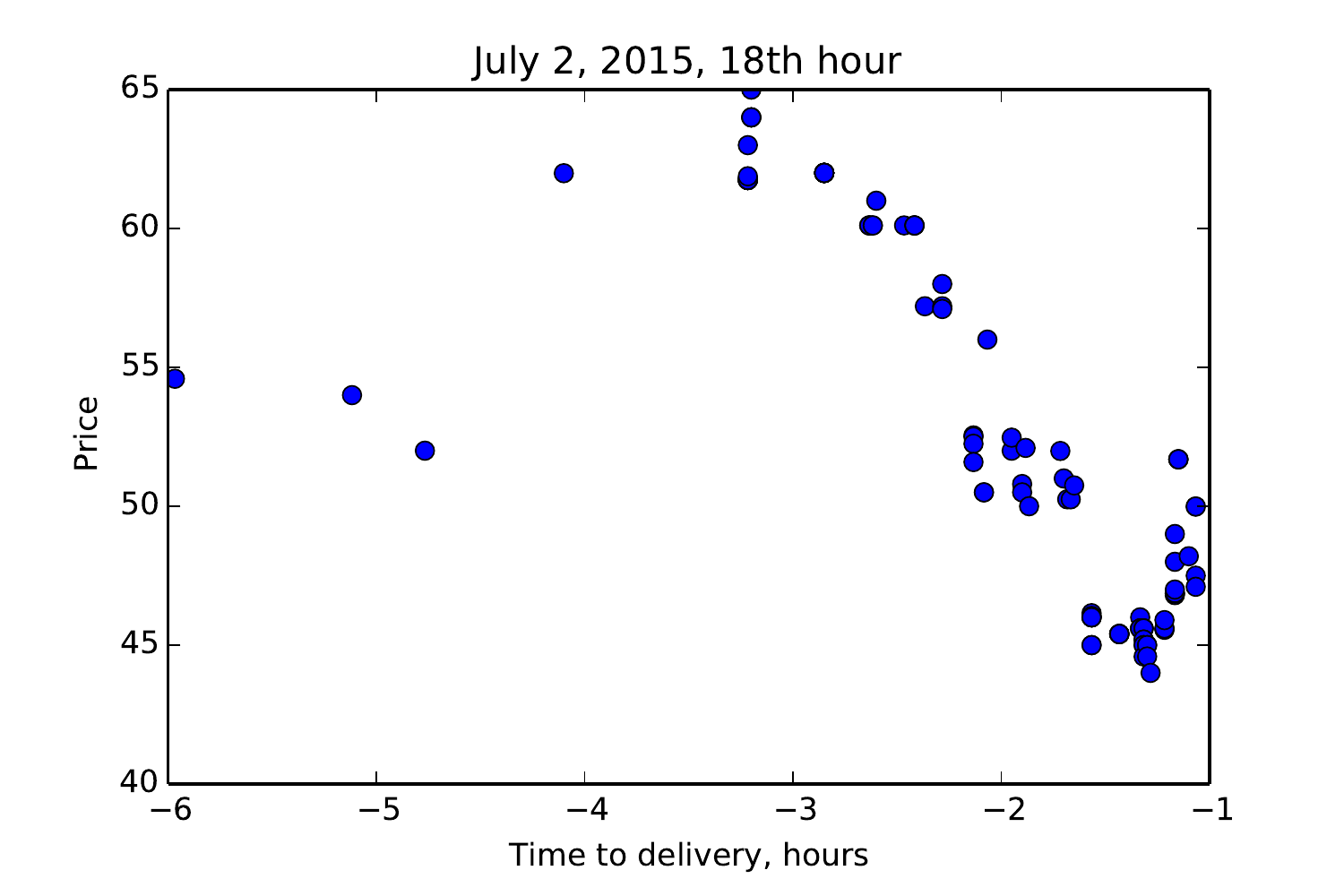}\includegraphics[width=0.5\textwidth]{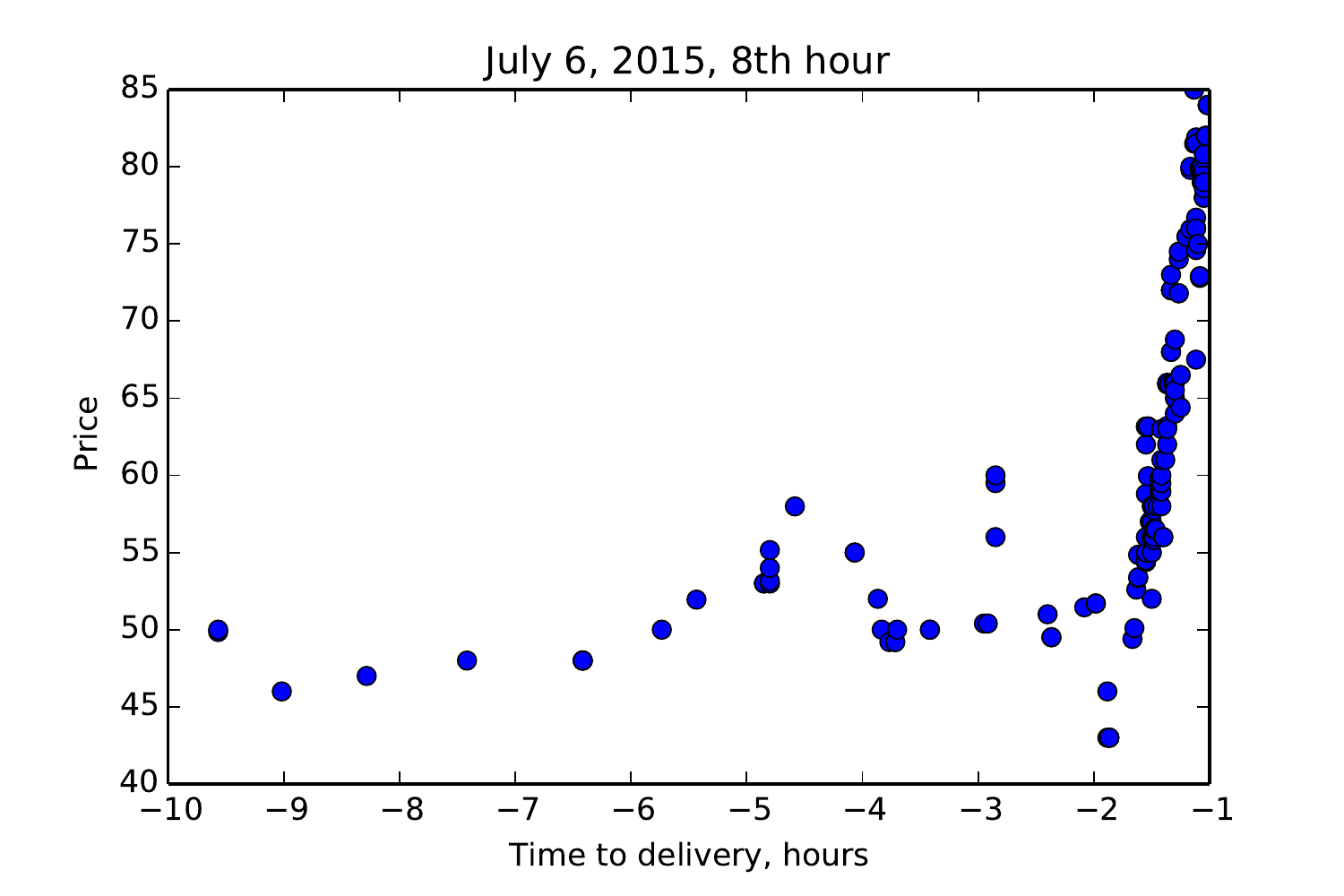}
}
\caption{Intraday transaction prices for a fixed delivery hour. As the delivery
horizon draws close, volatility increases and bid-ask spreads widen.}
\label{intraday.fig}
\end{figure}

\paragraph{Gain from trading}
For a small producer whose transactions do not create market impact, 
the total gain from selling
electricity is modeled (in continuous time) by 
\begin{align*}
G &= F_T P_T - \int_0^T \phi_t dP_t(T) - u(F_T - \phi_T),
\end{align*}
where $\phi \in \mathcal A$. 
For a large producer, with the
penalization by the market impact term, the gain from trading becomes
$$
G = F_T P_T - \int_0^T \phi_t dP_t(T) - u(F_T - \phi_T) -
\frac{\gamma}{2}\int_0^T \psi^2_s ds.
$$
where $\phi\in \mathcal A^{ac}$. In the following sections we consider
separately the problems of maximizing the expected gain with and
without market impact.


\subsection{Trading for a small producer in absence of market impact}
The optimization problem for trading in the absence of market impact writes
\begin{align}
\min_{\phi \in \mathcal A} \mathbb E\Bigg[\underbrace{\int_0^T \phi_t \mu_t
    dt}_{\text{\parbox{0.22\textwidth}{Expected loss from\\[-0.2cm]
        trading early (negative)}}} +   \underbrace{u(F_T -
  \phi_T)}_{\text{volume penalty}}\Bigg].\label{optprob}
\end{align}
Note that the stochastic part of
the price process does not play a role in this optimization
problem, and the solution is therefore independent from the price
volatility. One could introduce a risk penalty to account for the
price volatility effects but we do not pursue this here. 

Before obtaining a general solution with numerical methods, we first
present explicit solutions in the cases where the forecast information
is either exact or unavailable.  This will allow us to establish upper
and lower bounds on the expected gain which may be obtained with
probabilistic forecast. 

\paragraph{Exact forecast}
	Assume that the future realized production is known without
        error, in other words, $F_T \in \mathcal G_0$. 
In this case, the optimal strategy is described by the following
proposition, where we define
$$ t^* = \argmin_{0 \leq t \leq T} \int_t^T \mu_s ds, \qquad m^* =
m_{t^*} = \int_{t^*}^T \mu_s ds.$$
\begin{proposition}\label{exact.prop}
Let the penalty function $u$ be convex and continuously
differentiable, with $u'(0)=0$ and $\lim_{x\to-\infty} u'(x) =
-\infty$.
Then the
value function of the the problem \eqref{optprob} is given by 
\begin{align*}
&F_T m^*-v(m^*),\quad &&m^*<0;\\
&-v(0) \equiv  u(0),\quad &&\text{otherwise},
\end{align*}
where $v(y)$ is the Fenchel
transform of $u$: $v(y) = \sup_x \{xy - u(x)\}$. 

Denote by $I(y)$ the inverse function of $u'$. The optimal strategy for the problem \eqref{optprob} is described as
follows. 
\begin{itemize}
\item If $m^*\leq 0$, sell the quantity $\phi = F_T - I(m^*)$ at time
  $t^*$. 
\item If $m^*>0$, sell at time $T$ the amount $F_T$. 
\end{itemize}
\end{proposition}
\begin{remark}
Although the realized production is known in advance, sometimes it is
advantageous for the agent to sell more than the realized production,
to be able to benefit from the higher prices in the beginning of
trading. 
\end{remark}
\begin{proof}
We first transform the optimization functional with an integration by parts.
$$
\int_0^T \phi_t \mu_t dt + u(F_T - \phi_T)  = \int_0^T \left(\int_t^T
  \mu_s ds\right) d\phi_t + u(F_T - \phi_T).
$$
It is now clear that for fixed $\phi_T$, the optimal solution
$(\phi_t)_{0\leq t\leq T}$ is such that the measure $d\phi$ is
supported by the single point $t^*$. Therefore, $\phi_t = \phi\mathbf
1_{t\geq t^*}$ with a constant $\phi$. If $m^*>0$, it is optimal to
choose $t^* = T$ and $\phi$ which minimizes
$u(F_T-\phi)$, that is $\phi = F_T$.  Otherwise, $\phi$ can be found by solving
the optimization problem
$$
\min_{\phi\geq 0 }  \left[ \phi m^* + u(F_T - \phi) \right].
$$
The candidate optimizer is given by $\phi = F_T - I(m^*)$. It is
easy to check from our assumption that this quantity is positive,
which means that 
$$
\min_{\phi\geq 0 }  \left[ \phi m^* + u(F_T - \phi) \right] =
\min_{\phi\in \mathbb R }  \left[ \phi m^* + u(F_T - \phi) \right] =
v(m^*) + F_T m^*. 
$$
\end{proof}
\paragraph{Absence of forecast}
	In this case, we assume that the agent does not have access to
        the forecast but only knows the distribution of power
        production, in other words, $(\mathcal G_t)_{0\leq t< T}$ coincides
        with the completed natural filtration of $B$ and $\mathcal
        G_T$ in addition contains $F_T$. The agent's strategy is then deterministic on
        $[0,T)$ with a possible random jump at time $T$ (when the
        realized production becomes known). 
The optimal strategy is described by the following proposition, where
we define
	$ \tilde{u}(x) = \mathbb{E}\left[\bar u(F_T - \mathbb{E}[F_T]
          + x)\right]$ and $\bar u (x) = u(x)$ if $x\leq 0$ and $\bar u(x) = u(0)$ if
$x>0$. 
\begin{proposition}\label{noforecast.prop}
Let $\mathbb E[|F_T|]<\infty$ and assume that the penalty function $u$
satisfies the assumptions of Proposition \ref{exact.prop} and in addition 
$$
\mathbb E[u(F_T+x)]<\infty \quad \text{and} \quad \mathbb
E[|u'(F_T+x)|]<\infty \quad \forall x \in \mathbb R.
$$
Then the
value function of the the problem \eqref{optprob} is given by 
\begin{align*}
&\mathbb E[F_T] m^* - \tilde v(m^*),&& m^*<0;\\
& u(0), &&\text{otherwise,}
\end{align*}
where $\tilde v(y) = \sup_x\{xy - \tilde u(x)\}$.
The optimal strategy is described as
follows (we denote the inverse function of $\tilde u'$ by $\tilde I$). 
\begin{itemize}
\item If $m^*\leq 0$, sell the quantity $\phi = \mathbb E[F_T] -
  \tilde I(m^*)$ at time
  $t^*$ then sell $F_T - \mathbb E[F_T] +
  \tilde I(m^*)$ (if this quantity is positive) at time $T$. 
\item If $m^*>0$, sell the quantity $F_T$ at time $T$. 
\end{itemize}
\end{proposition}
\begin{proof}
Using the assumptions on $u$, by the dominated convergence theorem, we
can show that $\tilde u$ is convex and continuously differentiable. 
Similarly to the proof of Proposition \ref{exact.prop}, we find that
when $m^*<0$, 
the optimal strategy has the form 
$$
\phi \mathbf 1_{t\geq t^*} + (F_T - \phi)^+ \mathbf 1_{t\geq T},
$$
where $\phi$ is found by solving the optimization problem
$$
\min_{\phi\geq 0} \{\phi m^* + \tilde u(\mathbb E[F_T] - \phi)\}.
$$
The candidate optimizer is given by $\phi = \mathbb E[F_T] - \tilde
I(m^*)$. From our assumptions it follows that $\phi$ is nonnegative,
and therefore 
$$
\min_{\phi\geq 0} \{\phi m^* + \tilde u(\mathbb E[F_T] - \phi)\} =
\min_{\phi\in \mathbb R} \{\phi m^* + \tilde u(\mathbb E[F_T] - \phi)\}=\mathbb E[F_T] m^* - \tilde v(m^*).
$$
\end{proof}

	
\begin{example}	
Assume that the penalty is quadratic, that is, $u(x) =
\frac{\kappa}{2} x^2 $ and $\bar u(x) = \frac{\kappa}{2} x^2
\mathbf{1}_{x < 0}$, and the realized production $F_T$ is uniformly distributed on $[0,1]$. Then 
		\begin{equation*}
		\tilde{u}(x) = \frac{\kappa}{6}\left( \frac{1}{2} - x \right)^3 \mathbf{1}_{-\frac{1}{2} \leq x \leq \frac{1}{2}} + \frac{\kappa}{2} \left( x^2 + \frac{1}{12} \right) \mathbf{1}_{x < -\frac{1}{2}}
		\end{equation*}
		and thus
		\begin{equation*}
		\tilde{I}(z) = \left\{
		\begin{array}{ll}
		\frac{z}{\kappa} & z < -\frac{\kappa}{2} \\
		\frac{1}{2} - \sqrt{- \frac{2z}{\kappa} } & z \geq -\frac{\kappa}{2}
		\end{array}
		\right.
		\end{equation*}
		
\end{example}	

\paragraph{Discrete forecast updates}
In this paragraph we consider the more realistic situation when the
forecast is updated at a finite set of deterministic times $0=t_0<
t_1< \dots < t_n = T$, that is, 
$$F_t = \sum_{i=0}^{n-1} F_i \mathbf
1_{t_i \leq t< t_{i+1}} + F_n \mathbf 1_{t_n \leq t},
$$ 
where $(F_k)$ is a discrete-time martingale with respect to the discrete-time filtration $\mathcal F_k =
\sigma(F_i,0\leq i\leq k)$. 
 Moreover, we make the assumption that $\mu_t
\leq 0$ for $t\in [0,T]$, that is, the expected price may only fall as
the delivery date approaches. We denote 
$$
m_k = \int_{t_k}^{t_{k+1}} \mu_s ds. 
$$ 
The optimal strategy is described by the
following proposition. 
\begin{proposition}
Let the penalty function $u$ satisfy the assumptions of Proposition
\ref{noforecast.prop} and assume in addition that 
$$
\mathbb E[|u'(x-c(F_0+\dots+F_n))|]<\infty
$$ 
for all $x\in \mathbb R$ and some constant $c>1$. 
Then there exists a discrete-time $(\mathcal F_k)$-adapted process
$(\xi_k)_{0\leq k \leq n}$ such that 
\begin{align}
\sum_{i=k}^{n-1}m_i = \mathbb E[u'(F_n - \max_{k\leq i \leq n}
\xi_i)|\mathcal F_k]\label{bek.discrete}
\end{align}
for $k=0,\dots,n$. The optimal trading strategy is given by 
\begin{align}
\phi_t = \sum_{i=0}^{n-1} \phi_i \mathbf
1_{t_i \leq t< t_{i+1}} + \phi_n \mathbf 1_{t_n \leq t},\label{optstrat}
\end{align}
where 
$$
\phi_k = \max_{0\leq i\leq k} \xi_i
$$
for $0\leq k \leq n$. 
\end{proposition}
\begin{remark}
The process $(\xi_k)$ may be computed by backward induction. This
proposition can be extended to the continuous-time case using the
results of \cite{bank2004stochastic}, following, e.g.,
\cite{chiarolla2014identifying}. However, the discrete-time case is more
relevant in practice, since the forecasts are updated in discrete
time. In addition, for numerical computations time must be discretized
anyway. For this reason we concentrate on the discrete case in this
paper. 
\end{remark}
\begin{proof}
We first prove the existence of the process $(\xi_k)$ by an induction argument. Clearly, one
may choose $\xi_n = F_n$. Assume now that for some $m\leq n$, we have
constructed a process $(\xi_k)_{m\leq k\leq n}$ satisfying
\eqref{bek.discrete} for $m\leq k \leq n$, and such that in addition
\begin{align}
0\leq \xi_k \leq c F_k  - I \left(\frac{c}{c-1}\sum_{i=k}^{n-1}m_i\right) \label{xiest}
\end{align}
for $m\leq k \leq n$. 

Consider a random function
$$
\xi \mapsto f_{m}(\xi) = \mathbb E[u'(F_n - \max(\xi, \max_{m\leq i \leq n}
\xi_i))|\mathcal F_{m-1}]. 
$$
This function is well defined and a.s.~continuous for all $\xi \in \mathbb R$ by
assumptions of the proposition and estimate \eqref{xiest}. Remark that 
$$
\lim_{\xi \to \infty} f_m(\xi) = -\infty
$$
and by the induction hypothesis,
$$
f_m(0) = \sum_{i=m}^{n-1} m_i. 
$$
Therefore, there exists $\xi_m \in \mathcal F_m$ with $\xi_m\geq 0$, which solves the
equation $f_m(\xi) = \sum_{i=m-1}^{n-1}m_i$. Moreover, by Markov inequality,
\begin{align*}
&\mathbb E[u'(F_n - \max(\xi, \max_{m\leq i \leq n}
\xi_i))|\mathcal F_{m-1}]\leq \mathbb E[u'(F_n -
\max(\xi,F_n))|\mathcal F_{m-1}]\\
& \leq \mathbb E[u'(\min(F^n - \xi,0))\mathbf 1_{F_n \leq c \mathbb
  E[F_n|\mathcal F_{m-1}]}|\mathcal F_{m-1}]\\
& \leq \frac{c-1}{c} u' (\min(c\mathbb E[F_n|\mathcal F_{m-1}] -
\xi,0) ) \leq \frac{c-1}{c} u' (c\mathbb E[F_n|\mathcal F_{m-1}] -
\xi ).
\end{align*}
This shows that
$$
\xi_{m-1}\leq c\mathbb E[F_n|\mathcal F_{m-1}] - I \left(\frac{c}{c-1}\sum_{i=m-1}^{n-1}m_i\right).
$$

It remains to show the optimality of the proposed strategy. First,
note that since $\mu_t\leq 0$ for all $t\in [0,T]$, with each interval
it is optimal to sell electricity as early as possible, so that the
optimal strategy has the form \eqref{optstrat}. Therefore, we need to
minimize the discrete-time version of the objective function
$$
 J(\phi) =\mathbb E\left[\sum_{k=0}^{n-1} \phi_k m_k + u(F_n -
   \phi_n)\right]
$$
over all increasing adapted discrete-time processes $(\phi_k)_{0\leq k \leq
  n}$. Let $\phi^*_k = \max_{0\leq i\leq k}\xi_i$ and let $(\phi_k)$ be
any other admissible strategy. We denote $\Delta \phi_i = \phi_i -
\phi_{i-1}$ for $i=1,\dots,n$ and $\Delta \phi_0 = \phi_0$, and
similarly for $\Delta \phi^*_i$. Then,
\begin{align*}
J(\phi) - J(\phi^*)  &= \mathbb E\left[\sum_{k=0}^{n-1} (\phi_k -
  \phi^*_k) m_k + u(F_n-\phi_n) - u(F_n - \phi^*_n)\right]\\
&\geq \mathbb E\left[\sum_{k=0}^{n-1} (\phi_k -
  \phi^*_k) m_k - u'(F_n - \phi^*_n)(\phi_n -
  \phi^*_n)\right]\\
& = \mathbb E\left[\sum_{i=0}^n (\Delta \phi_i - \Delta \phi^*_i)
  \mathbb E\left[\sum_{k=i}^{n-1} m_k - u'(F_n - \phi^*_n)|\mathcal F_k\right]\right].
\end{align*}
Now, on the one hand, for all $k$, $\phi^*_n \geq \max_{k\leq i \leq
  n} \xi_i$, which means that 
$$
\mathbb E\left[\sum_{k=i}^{n-1} m_k - u'(F_n - \phi^*_n)|\mathcal
  F_k\right]\leq \mathbb E\left[\sum_{k=i}^{n-1} m_k - u'(F_n - \max_{k\leq i \leq
  n} \xi_i)|\mathcal
  F_k\right]=0.
$$
On the other hand, if, for some $k$, $\Delta \phi_k > 0 $ then $\phi_k
= \xi_k$ so that 
$$
\mathbb E\left[\sum_{k=i}^{n-1} m_k - u'(F_n - \phi^*_n)|\mathcal
  F_k\right] = \mathbb E\left[\sum_{k=i}^{n-1} m_k - u'(F_n - \max_{k\leq i \leq
  n} \xi_i)|\mathcal
  F_k\right] = 0.
$$
Since the processes $\phi$ and $\phi^*$ are nondecreasing, these
observations together with the above estimate imply that 
$$
J(\phi) - J(\phi^*)\geq 0,
$$
which means that $\phi^*$ is the optimal strategy. 

\end{proof}
\subsection{Trading for a large producer in presence of market impact}\label{large.sec}
Our aim now is to maximize the expected gain penalized by market
impact and volume penalty. 
The optimization problem therefore takes the
following form. 
\begin{align}
\min_{\phi \in \mathcal A^{ac}_+} \mathbb E\Bigg[\underbrace{\int_0^T \phi_t \mu_t
    dt}_{\text{\parbox{0.22\textwidth}{Expected loss from\\[-0.2cm]
        trading early (negative)}}} +
  \underbrace{\frac{\gamma}{2}\int_0^T \psi^2_s ds}_{\text{\parbox{0.13\textwidth}{Market
      impact\\[-0.2cm]($\psi = \phi'$)}}} 
+   \underbrace{u(F_T -
  \phi_T)}_{\text{volume penalty}}\Bigg]. \label{optprob.impact}
\end{align}
When buying electricity is allowed, the set $\mathcal A^{ac}_+$ is
replaced with the set $\mathcal A^{ac}$. 

We first consider the situation when the agent is only allowed to sell
electricity. 
As before, we first focus on the degenerate cases when the forecast is
either exact or unavailable. 
\begin{proposition}[Exact forecast]\label{mi.exact}
Let the penalty function $u$ be strictly convex and continuously
differentiable with $u'(0) = 0$ and $\lim_{x\to -\infty} u'(x) =
-\infty$. Then the optimal strategy for the problem
\eqref{optprob.impact} is given by 
$$ 
\bar{\psi_t} = \frac{1}{\gamma} \left( u'(F_T - \bar{\phi}_T) - \int_t^T \mu_s ds \right)^+,	
$$
where the terminal value $\bar{\phi}_T$ is the solution of the equation
$$ \phi = \frac{1}{\gamma} \int_0^T dt \left( u'(F_T - \phi) - \int_t^T \mu_s ds \right)^+. $$
\end{proposition}
\begin{proof}
The Hamiltonian of this optimization problem is
	$$ H(t, \phi, x , \psi ; p) = \left(\mu_t \phi + \frac{\gamma}{2} \psi^2 \right) + p \psi$$ 
	By Pontriagin's principle for deterministic control problems,
        for each $t$, the optimal strategy $\bar{\psi_t}$ realizes the
        minimum of $H(t,\bar{\phi_t}, x, \psi; \bar{p_t}) = \bar{p_t}
        \psi + \frac{\gamma}{2}\psi^2 + \bar{\phi_t} \mu_t$ over all
        $\psi\geq0$, where $\bar{\phi_t} = \int_0^t \bar{\psi_s} ds$ and the function $\bar{p_t}$ satisfies
	$$ \frac{d}{dt} \bar{p_t} = - \frac{\partial}{\partial \phi} H(t,\bar{\phi_t}, x, \bar{\psi_t}; \bar{p_t}) =  - \mu_t, \qquad p_T = - u'(F_T - \bar{\phi}_T)$$
	
	Therefore, 
	$$ \bar{p_t} = - u'(F_T - \bar{\phi}_T) + \int_t^T \mu_s ds $$
	and finally
	$$ \bar{\psi_t} = \argmin_{\psi >0} H(t, \bar{\phi_t}, x, \psi ; \bar{p_t}) = - \frac{1}{\gamma} ( \bar{p_t} \wedge 0) = \frac{1}{\gamma} \left( u'(F_T - \bar{\phi}_T) - \int_t^T \mu_s ds \right)^+	
	$$
	where the terminal value $\bar{\phi}_T$ is the solution of the equation
	\begin{align} \phi = \frac{1}{\gamma} \int_0^T dt \left( u'(F_T - \phi) -
          \int_t^T \mu_s ds \right)^+ .\label{termeq}\end{align} It is easy to see that under
        our assumptions, this equation admits a unique solution which
        is strictly positive. 
\end{proof}

\begin{example}
	To obtain an explicit solution, assume that $\mu_s \equiv
        \mu<0$ is constant, and that the penalty function is
        quadratic: $u(x) = \frac{\kappa}{2} x^2$. A straightforward
        computation then gives the solution to equation
        \eqref{termeq}. 
	\begin{equation}
	\bar{\phi}_T = \left\{
	\begin{array}{ll}
	\displaystyle F_T - \frac{\frac{\mu T^2}{2} + \gamma F_T }{\kappa T + \gamma}, & \displaystyle \frac{\mu T^2}{2} + \gamma F_T  \geq 0 \\
	\displaystyle F_T - \frac{\mu}{\kappa^2}(\gamma+ \kappa T) - \sqrt{ \frac{\mu^2}{\kappa^4}(\gamma+ \kappa T)^2 - \frac{2\mu}{\kappa^2} \left( \frac{\mu T^2}{2} + \gamma F_T\right) }, & \displaystyle \frac{\mu T^2}{2} + \gamma F_T  < 0 
	\end{array}
	\right.
	\end{equation}

In the first case, $\bar{\phi}_T \leq F_T$ and the optimal trading
strategy satisfies
		$$\bar{\psi_t} = \frac{-\frac{1}{2}\kappa \mu T^2 + \gamma \kappa F_T - \mu \gamma T}{\gamma\kappa T + \gamma^2} + \frac{\mu}{\gamma} t \geq 0, \quad \forall t \in [0,T].$$ 
		In other words, the agent trades continuously between
                time $t = 0$ and $t=T$, at a linearly decreasing
                rate. The expected gain of the power producer is
			$$
			G_{exact}(T) = \mathbb{E}[P_T] F_T -\left[  \frac{\gamma \lambda}{2} T + \mu \lambda T^2 + \frac{\mu^2}{3 \gamma} T^3 + \frac{\kappa}{2} \left( \frac{\frac{\mu T^2}{2} + \gamma F_T }{\kappa T + \gamma}  \right)^2 \right] 
			$$
			where $\lambda = \frac{-\frac{1}{2}\kappa \mu T^2 + \gamma \kappa F_T - \mu \gamma T}{\gamma\kappa T + \gamma^2} $.
			
In the second case, we have $\bar{\phi}_T > F_T$. Introduce the time 
		\begin{equation}
			t^* =  T -  \left( T +\frac{\gamma}{\kappa} \right) + \sqrt{  \left( T +\frac{\gamma}{\kappa} \right)^2 - \frac{2}{\mu} \left( \frac{\mu T^2}{2} + \gamma F_T \right)}.
		\end{equation}
The optimal strategy is given by
		$$ \bar{\psi_t} = \frac{\mu \gamma + \sqrt{ \mu^2
                    \gamma^2 - 2 \mu \kappa \gamma ( \kappa F_T - \mu
                    T) } } {\gamma \kappa}  + \frac{\mu}{\gamma} t
                , $$
for $t \in [0, t^*]$ and $\bar \psi_t = 0$ for $t>t^*$. In other
words, the agent trades continuously at a linearly decreasing rate
until time $t^* < T$ and then stops. The expected gain of the power producer is
		$$
		G_{exact}(T) =\mathbb{E}[P_T]. F_T - \left[  \frac{\gamma \lambda}{2} t^* + \mu \lambda (t^*)^2 + \frac{\mu^2}{3 \gamma} (t^*)^3 + + \mu (T - t^*) \bar{\phi}_T + \frac{\kappa}{2} \left( F_T - \bar{\phi}_T \right)^2 \right] 
		$$
		where 
		$ \lambda  = \frac{\mu \gamma + \sqrt{ \mu^2 \gamma^2 - 2 \mu \kappa \gamma ( \kappa F_T - \mu T) } } {\gamma \kappa}  $

\end{example}
\begin{remark}
The solution in the case when no forecast is available can be obtained
by replacing the penalty function $u$ with the average penalty
$$ \tilde{u}(x) = \mathbb{E}\left[ u(F_T - \mathbb{E}[F_T] + x)
\right] $$
in Proposition \ref{mi.exact}. The strategy is completely
deterministic in this case since contrary to the situation without
market impact, there is no lump-sum trade at the terminal date. 
\end{remark}

\paragraph{Continuous forecast updates} In the presence of market
impact, since the trading strategy is necessarily continuous-time, we
assume that the forecast is updated in continuous time as well. 
To solve this problem, introduce the value function
$$
w(t,\phi,x) = \min_{\psi \geq 0} \mathbb E\left[\int_t^T \phi_s \mu_s ds +
  \frac{\gamma}{2}\int_t^T \psi^2_s ds +  u(F_T -
  \phi_T)\Big| \phi_t = \phi, X_t = x\right].
$$
The following proposition can be obtained using standard tools of
stochastic control (see Propositions 4.3.1 and 4.3.2 in
\cite{pham2009continuous}). These standard results do not guarantee the
uniqueness of the viscosity solution, because the set of controls is
not bounded, but if the strategy $\psi$ is restricted to a bounded
domain, uniqueness follows easily from Theorem 4.4.5 in \cite{pham2009continuous}.
\begin{proposition}\label{prophjb}
The value function $w(t,\phi,x)$ is a viscosity solution of
the Hamilton-Jacobi-Bellman equation
$$
\min_{\psi\geq 0}\left\{ \frac{\gamma}{2} \psi^2 +
\frac{\partial w}{\partial \phi} \psi \right\} + \phi \mu_t + \frac{\partial w}{\partial t}+ \frac{1}{2} \sigma_t^2 x^2
\frac{\partial^2 w}{\partial x^2} = 0
$$
or equivalently,
$$
- \frac{1}{2\gamma} \left(\frac{\partial w}{\partial \phi} \wedge 0\right)^2  + \phi \mu_t + \frac{\partial w}{\partial t}+ \frac{1}{2} \sigma_t^2 x^2
\frac{\partial^2 w}{\partial x^2} = 0
$$
for $\phi\geq 0$, $x\geq 0$ and $t\in [0,T]$ with the terminal condition 
$$
w(T,\phi,x) = u(f_{prod}(x) - \phi). 
$$
\end{proposition}

\paragraph{The case when both buy and sell transactions are allowed}
In the final paragraph we consider the situation when the agent can
both buy and sell electricity. In this case, for the quadratic penalty
function, the optimal strategy is
explicit (since $\mathcal A^{ac}$ is a linear space) and described by
the following proposition. For the non-quadratic penalty function the
optimal strategy may be found by solving the HJB equation as above.  
\begin{proposition}\label{propl2}
Assume that the penalty function is quadratic: $u(x) =
\frac{x^2}{2}$. Then the optimal trading rate satisfies
$$
\psi^*_t = \frac{\mathbb E[F_T|\mathcal G_t] - \phi^*_t -
  \frac{1}{\gamma}\int_t^T ds\, (\gamma + T-s) \mu_s ds}{\gamma + T-t}.
$$
\end{proposition}
\begin{proof}
The first order condition writes
$$
\mathbb E\left[\int_0^T dt\, \xi_t\left(\int_t^T \mu_s ds
  +\gamma \psi_t  - (F_T -
  \phi_T) \right)
\right]=0
$$
for every $\mathbb G$-adapted process $\xi$. Therefore,
\begin{align}
\gamma \psi_t + \int_t^T \mu_s ds =  \mathbb E\left[
     F_T -
  \phi_T\Big|\mathcal G_t\right], \label{1storder}
\end{align}
which means that the left-hand side is a martingale. This in turn
means that for $s\geq t$, 
$$
\mathbb E[\psi_s|\mathcal G_t ] = \psi_t +
\frac{1}{\gamma}\int_t^s \mu_u du. 
$$
Substituting this formula into \eqref{1storder}, we then obtain
\begin{align*}
(\gamma+T-t) \psi_t  =  
     \mathbb E[F_T|\mathcal G_t] -
  \phi_t - \frac{1}{\gamma}\int_t^T ds\, (\gamma + T-s) \mu_s ds.
\end{align*}
\end{proof}

\subsection{Numerical illustrations}

In this section, we illustrate the optimal trading policies for a
large producer, determined in section \ref{large.sec}, with numerical examples.

In these examples, we assume that the trading takes place
continuously over $T=6$ days, that $\mu= -0.2$ (this means that the forward
price per MWh decreases by 1 euro every 5 days as one approaches
maturity), $\gamma = 4800$ (liquidating 0.1MWh over 1 hour has a cost
of approximately 1 euro), the daily volatility of $(X_t)$ is $\sigma_t
\approx 27\%$ (that is, $66\%$ over the 6 days; this corresponds roughly to the estimated value for one
of the power plants we studied in this paper), and that the penalty
function is $u(x) = P x^2$ with $P=100$ (this means that with, e.g.,
0.1MWh volume mismatch, the extra price to pay is 1 euro).  

To obtain the numerical examples, we first solve the HJB equation by finite
differences, and then simulate random trajectories of the forecast
process with a fixed value of realized production. For these
trajectories, we compute the corresponding trajectories of the optimal
trading strategy $\phi_t$. Figure \ref{strategies.fig} shows several
sample trajectories when only selling is allowed (on the left graph)
and when both buy and sell transactions are permitted (right graph). 

\begin{figure}
\centerline{\includegraphics[width=0.55\textwidth]{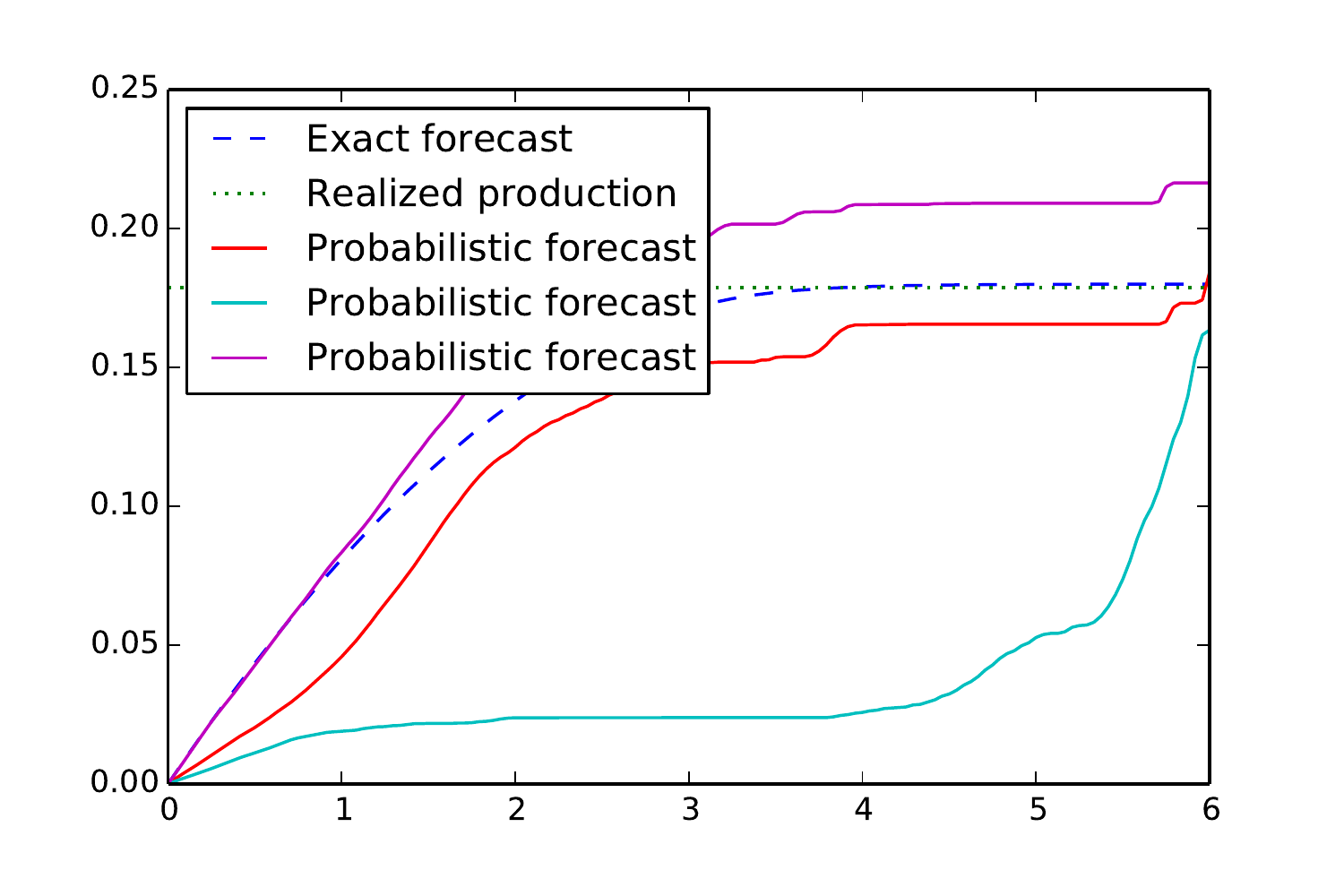}\includegraphics[width=0.55\textwidth]{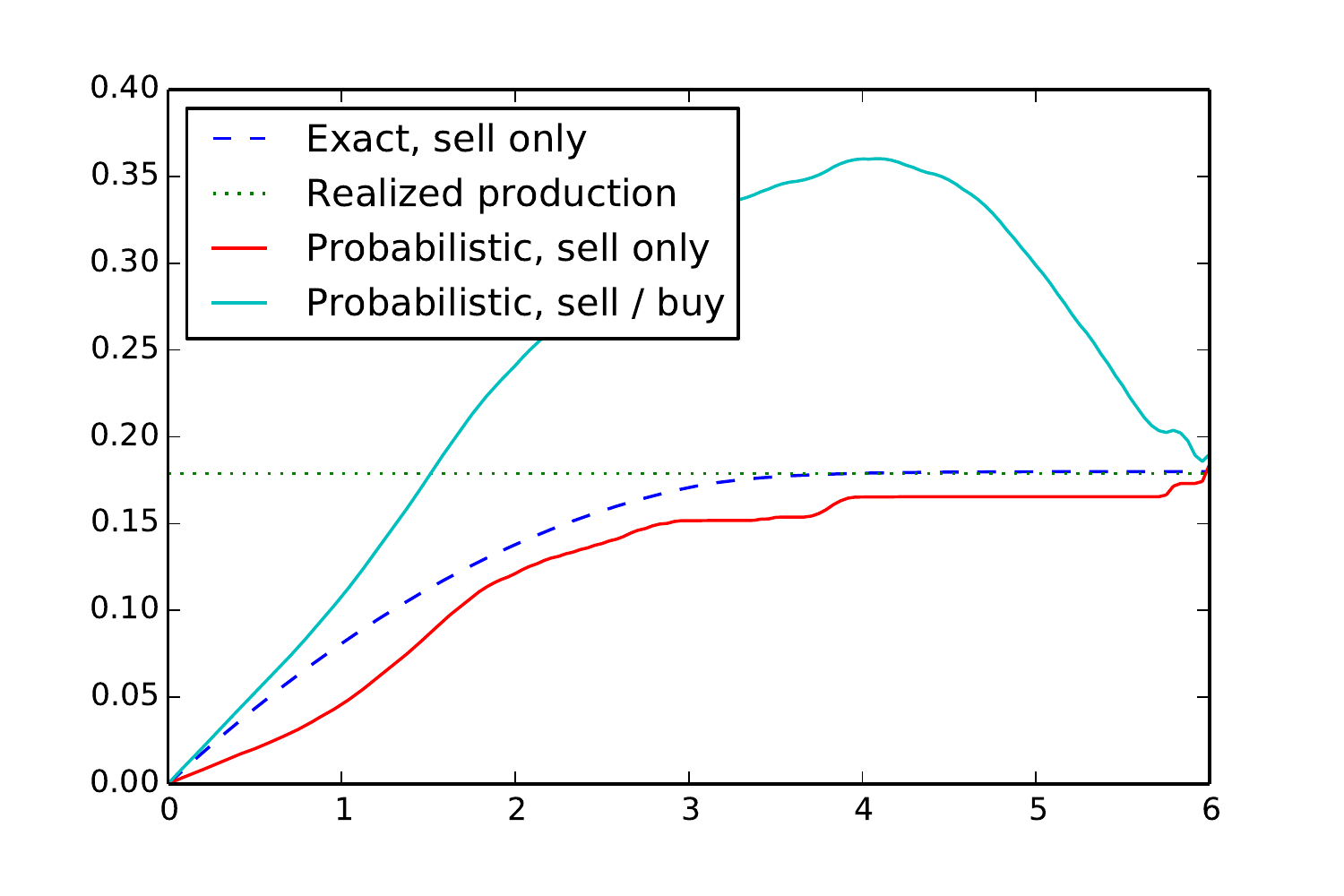}}

\caption{Sample selling strategies with market impact. Strategies are updated dynamically as new information
becomes available. The left graph shows three strategies where sales
only are allowed. The right graph shows one strategy with sales only
and one with both buy and sell transactions. }
\label{strategies.fig}
\end{figure}

Finally, we compute the realized penalty (the value of the expression
under the expectation in \eqref{optprob.impact}) corresponding to the
simulated trajectories of the forecast process and the optimal trading
strategy, with the objective of evaluating the economic value of the
optimal strategy in different contexts. Figure \ref{quality.fig}, left
graph compares the distribution of the realized penalty with
volatility $\sigma = 66\%$  and that with volatility $33\%$. One can
see that with the lower volatility, the premium for early trading
compensates the cost of market impact and the volume penalty, leading
to negative overall penalty for most
of the trajectories, whereas for the higher volatility, the penalty is
positive for most trajectories. Figure \ref{quality.fig}, right graph,
quantifies the impact of allowing both buy and sell transactions (the
volatility was taken to be $66\%$ for both experiments). One
can see that once again, if the agent is allowed to both buy and sell,
the premium for early trading compensates the volume penalty and the
cost of market impact.

\begin{figure}
\centerline{\includegraphics[width=0.55\textwidth]{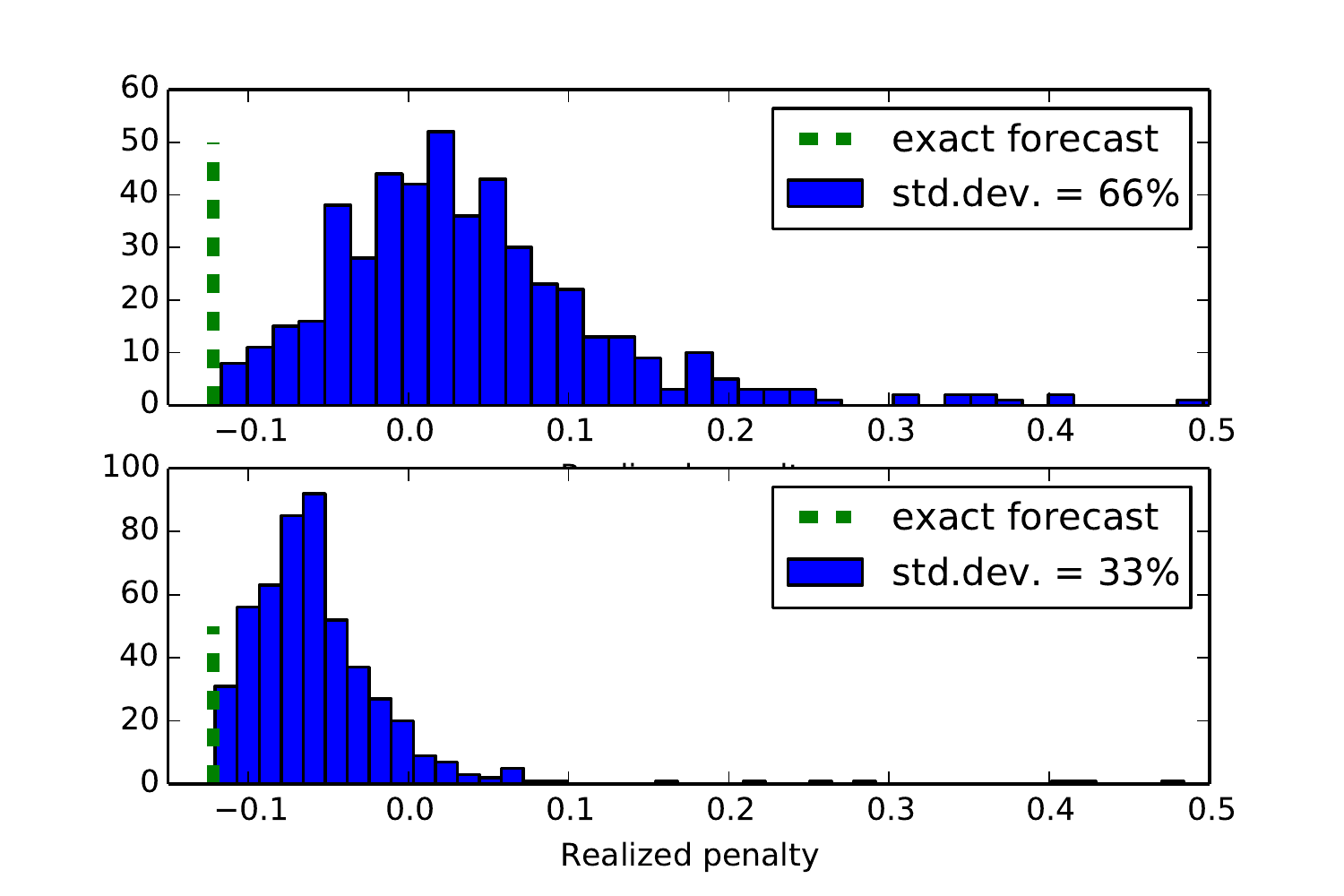}\includegraphics[width=0.55\textwidth]{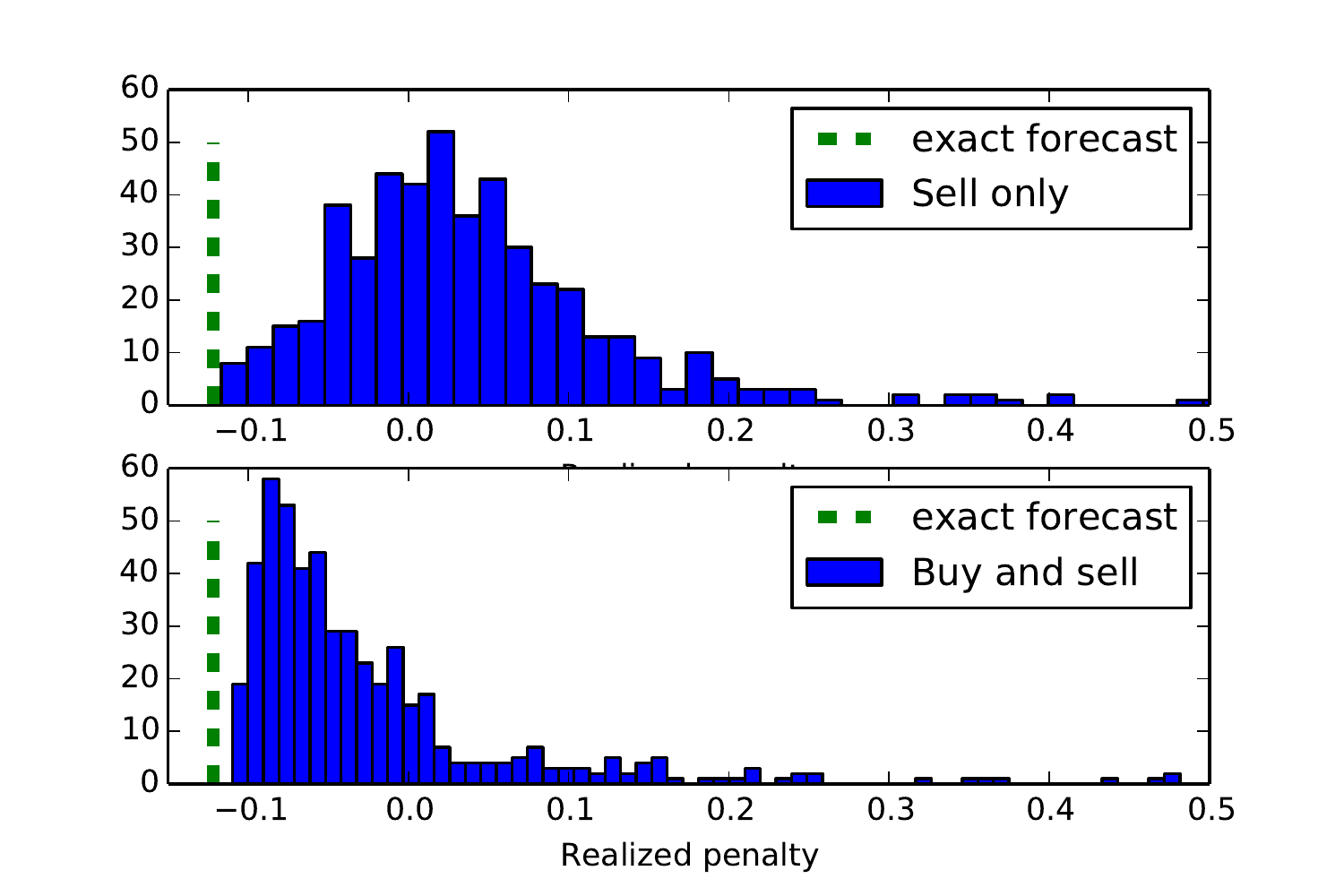}}

\caption{Left: Realized penalty for different forecast quality. Right:
realized penalty with and without buy transactions.} 
\label{quality.fig}
\end{figure}

\section*{Acknowledgement} We are grateful to Nicolas Girard and
Sophie Guignard from Ma\"ia Eolis and to J\'er\^ome Collet and Olivier
Feron from EDF Lab for helpful
discussions and for providing the data.
This work is supported by the French National
Research Agency (ANR) as part of the project Forewer (ANR-
14-CE05-0028).


\end{document}